\documentclass[journal,draftcls,onecolumn,12pt,twoside]{IEEEtranTCOM}

\usepackage{comment}
\usepackage{cite}
\usepackage{graphicx}
\usepackage{epstopdf}
\usepackage{amsthm}
\usepackage[cmex10]{amsmath}
\usepackage{array}
\usepackage{algorithm}
\usepackage[noend]{algpseudocode}
\usepackage{float}
\usepackage{setspace}
\usepackage{graphicx}
\usepackage{amssymb}
\usepackage{microtype}
\usepackage{color}
\usepackage{relsize}
\usepackage[font=footnotesize]{caption}
\usepackage{multirow}
\usepackage{balance}
\usepackage{amsmath}
\usepackage{hhline}
\usepackage[caption=false]{subfig}
\usepackage{flushend}
\usepackage{bm}
\usepackage{cancel}
\usepackage[normalem]{ulem}
\let\labelindent\relax
\usepackage{enumitem}
\usepackage{todonotes}

\newtheorem{theorem}{Theorem}

\newtheorem{definition}{Definition}
\newtheorem{remark}{Remark}
\newtheorem{example}{Example}

\bstctlcite{IEEEexample:BSTcontrol}

\makeatletter
\def\BState{\State\hskip-\ALG@thistlm}
\makeatother

\def\defopp{\mathrel{\ooalign{%
\raisebox{2\height}{\scriptsize $p$}\cr\hidewidth$\neq$\hidewidth\cr}}}

\DeclareMathOperator*{\argmin}{arg\,min}

\normalsize

\begin{document}

\title{Multi-Dimensional Spatially-Coupled Code Design: Enhancing the Cycle Properties}

\author{Homa~Esfahanizadeh,~\IEEEmembership{Student Member,~IEEE}, Lev Tauz,~\IEEEmembership{Student Member,~IEEE}, and~Lara~Dolecek,~\IEEEmembership{Senior~Member,~IEEE}\vspace{-1cm}
\date{}
\markboth{IEEE Transactions on Communications}{Submitted paper}
\thanks{H. Esfahanizadeh, L. Tauz, and L. Dolecek are with the Department of Electrical and Computer Engineering, University of California, Los Angeles, Los Angeles, CA 90095 USA (e-mail: {hesfahanizadeh}@ucla.edu; {levtauz}@ucla.edu; dolecek@ee.ucla.edu).\vspace{0cm}}
\thanks{Parts of the paper were presented at the 56th Annual Allerton Conference on Communication, Control, and Computing 2018 \cite{EsfahanizadehAllerton2018}, and the 10th Annual Non-Volatile Memories Workshop 2019 \cite{EsfahanizadehNVMW2019}.}
}

\maketitle

\begin{abstract}
A circulant-based spatially-coupled (SC) code is constructed by partitioning the circulants in the parity-check matrix of a block code into several components and piecing copies of these components in a diagonal structure. By connecting several SC codes, multi-dimensional SC (MD-SC) codes are constructed. In this paper, we present a systematic framework for constructing MD-SC codes with notably better cycle properties than their one-dimensional counterparts. In our framework, the multi-dimensional coupling is performed via an informed relocation of problematic circulants. This work is general in the terms of the number of constituent SC codes that are connected together, the number of neighboring SC codes that each constituent SC code is connected to, and the length of the cycles whose populations we aim to reduce. Finally, we present a decoding algorithm that utilizes the structures of the MD-SC code to achieve lower decoding latency.
Compared to the conventional SC codes, our MD-SC codes have a notably lower population of small cycles, and a dramatic BER improvement. The results of this work can be particularly beneficial in data storage systems, e.g., 2D magnetic recording and 3D Flash systems, as high-performance MD-SC codes are robust against various channel impairments and non-uniformity.
\vspace{0cm}
\end{abstract}

\IEEEpeerreviewmaketitle
\newpage
\section{Introduction\vspace{0cm}}

Spatially-coupled (SC) codes are a family of graph-based codes that have attracted significant attention thanks to their capacity approaching performance. {One-dimensional SC (1D-SC)} codes are constructed by coupling  a series of disjoint block codes into a single coupled chain \cite{FelstromIT1999}. Here, we use circulant-based (CB) LDPC codes \cite{TannerIT2004} as the underlying block codes. {
The 1D-SC codes have been well studied from the asymptotic perspective and the finite length perspective. From the asymptotic perspective, density evolution techniques have been used to study the decoding threshold, e.g., \cite{LentmaierIT2010,KudekarIT2013}. From the finite length perspective, via the evaluation and optimization of the number of problematic combinatorial objects, it has been shown how an informed coupling strategy can notably improve the performance, e.g., see \cite{EsfahanizadehTCOM2018,MitchellISIT2017,BeemerISITA2016}.}

Multi-dimensional SC (MD-SC) codes can be constructed by coupling several {blue}{1D-SC} codes together via rewiring the existing connections or by adding extra variable nodes (VNs) or check nodes (CNs){\cite{TruhachevITA2012,OhashiISIT2013}}. MD-SC codes are more robust against burst erasures and channel non-uniformity, and they have improved iterative decoding thresholds, compared to 1D-SC codes{\cite{TruhachevITA2012,OhashiISIT2013}}. MD-SC codes were introduced in \cite{TruhachevITA2012,OhashiISIT2013} and investigated more in \cite{TruhachevICC2012,OlmosITW2013,SchmalenISTC2014,LiuCOMML2015,TanakaWCNC2017,AliWCL2018}. 

In  \cite{TruhachevITA2012,TruhachevICC2012,OlmosITW2013}, constructions are presented for MD-SC codes that have specific structures, e.g., loops and triangles.
The construction method for MD-SC codes presented in \cite{OhashiISIT2013} involves connecting edges uniformly at random such that some criteria on the number of connections are satisfied.
In \cite{SchmalenISTC2014}, a framework is presented for constructing MD-SC codes by randomly and sparsely introducing additional CNs to connect VNs at the same positions of different chains. 
In \cite{LiuCOMML2015}, multiple SC codes are connected by random edge exchanges between adjacent chains to improve the iterative decoding threshold. In \cite{TanakaWCNC2017,AliWCL2018}, MD-SC codes are presented to improve the error correction performance against the severe burst errors in wireless channels.

Previous works on MD-SC codes, while promising, have some limitations. In particular, they either consider random constructions or are limited to specific topologies. As a result, they do not effectively use the added degree of freedom achieved by the multi-dimensional (MD) coupling in order to improve particular properties of the code, e.g., girth and minimum distance. They also use the density evolution technique for the performance analysis. This technique is dedicated to the asymptotic regime and is based on some assumptions, e.g., being cycle-free, that cannot be readily translated to the practical finite-length case. In \cite{OlmosITW2013}, a finite-length analysis in the waterfall region for MD-SC codes with a loop structure is presented. 

Finding the best connections to be rewired in order to connect constituent {blue}{1D-SC} codes and construct MD-SC codes with outstanding finite-length performance is still an open problem. 
This paper is the first work to present a comprehensive systematic framework for constructing MD-SC codes by coupling individual SC codes together to attain fewer short cycles. For connecting the constituent SC codes, we do not add extra VNs or CNs, and we only rewire some existing connections.
This paper is an extended version of our work published in \cite{EsfahanizadehAllerton2018}. We extend our previous work by: (1) connecting an arbitrary number of SC codes at a desired MD coupling depth to construct MD-SC codes; (2) converting the instances of the short cycles in the constituent SC codes to cycles of the largest possible length in the MD-SC code; and (3) presenting a low-latency decoder that exploits the structure of the constituent SC codes along with the structure of the final MD-SC code.

For exchanging the connections, we follow three rules: (1) The connections that are involved in the highest number of short cycles are targeted for rewiring; (2) The neighboring constituent SC codes to which the targeted connections are rewired are chosen such that the associated short cycles convert to cycles of the largest possible length in the MD setting; (3) The targeted connections are rewired to the same positions in the other constituent SC codes in order to preserve the low-latency decoding property. From an algebraic viewpoint, problematic circulants (which correspond to groups of connections) that contribute to the highest number of short cycles in the constituent SC codes are relocated to connect these codes together.

The rest of the paper is organized as follows. In Section II, the necessary preliminaries are briefly reviewed. In Section III, the structure of our MD-SC codes is presented. In Section IV, our novel framework for constructing MD-SC codes with enhanced cycle properties is introduced. In Section V, a low-latency algorithm for decoding MD-SC codes is presented. In Section VI, our simulation results are given. Finally, the conclusion appears in Section VII.\vspace{0cm}

\section{Preliminaries\vspace{0cm}}

Throughout this paper, each column (resp., row) in a parity-check matrix corresponds to a VN (resp., CN) in the equivalent graph of the matrix. Regular CB codes are $(\gamma,\kappa)$ LDPC codes, where $\gamma$ is the column weight of the parity-check matrix (VN degree), and $\kappa$ is the row weight (CN degree). The parity-check matrix $\mathbf{H}$ of a CB code is constructed as follows\vspace{0cm}:

\begin{equation}\label{CB_structure}
\mathbf{H}=\left[
\begin{array}{cccccccc}
\sigma^{f_{0,0}} & \sigma^{f_{0,1}}  & \dots & \sigma^{f_{0,\kappa-1}} \\
\sigma^{f_{1,0}} & \sigma^{f_{1,1}}  & \dots & \sigma^{f_{1,\kappa-1}} \\
\vdots & \vdots & \dots & \vdots\\
\sigma^{f_{\gamma-1,0}} & \sigma^{f_{\gamma-1,1}}  & \dots & \sigma^{f_{\gamma-1,\kappa-1}}
\end{array}
\right].
\end{equation}

Each circulant has the form $\sigma^{f_{i,j}}$ where $i$, $0\leq i\leq \gamma {-}1$, is the row group index, $j$, $0\leq j\leq \kappa {-} 1$, is the column group index, and $\sigma$ is the $z\times z$ identity matrix cyclically shifted one unit to the left. {blue}{The term $f_{i,j}$ specifies the power of the circulant at row group index $i$ and column group index $j$.} We use CB codes as the underlying block codes of SC codes.
{We highlight that, in this paper, each circulant in (\ref{CB_structure}) is a permutation of an identity matrix. Thus, each circulant has weight $1$. Circulants with larger weights have a negative impact on the girth \cite{WangIT2013}, and we do not use them in our code construction since the ultimate goal is to improve the cycle properties.
}

The parity-check matrix $\mathbf{H}_\textnormal{SC}$ of a CB SC code is constructed by partitioning the $\kappa\gamma$ circulants of the underlying block code into ($m+1$) component matrices $\mathbf{H}_0,\mathbf{H}_1,\dots,\mathbf{H}_m$ (with the same size as $\mathbf{H}$), and piecing $L$ copies of the component matrices together as shown in Fig.~\ref{H_SC_structure}. {blue}{The parameter $m$ is called the memory, and the parameter $L$ is called the coupling length.} Each component matrix $\mathbf{H}_l$, $0\leq l \leq m$, has a subset of circulants of $\mathbf{H}$ and zeros elsewhere so that $\sum_{l=0}^{m}\mathbf{H}_l=\mathbf{H}$. A replica $\mathbf{R}_\nu$, $1\leq \nu \leq L$, is a submatrix of $\mathbf{H}_\textnormal{SC}$ that has one submatrix $[\mathbf{H}_0^T\dots\mathbf{H}_m^T]^T$, Fig.~\ref{H_SC_structure}.

Recently, a systematic framework for partitioning the underlying block code and optimizing the circulant powers, known as the optimum partitioning and circulant power optimizer (OO-CPO) technique, was proposed for constructing high-performance SC codes \cite{EsfahanizadehTCOM2018,HareedyGC2018}. In this paper, we use the OO-CPO technique for designing the constituent SC codes that are then used to construct MD-SC codes. We note that choosing high-performance 1D-SC codes as constituent SC codes is not necessary in our MD-SC construction, and it only results in a better start point in a framework that further improves the performance via MD coupling.

Short cycles have a negative impact on the performance under iterative decoding. They affect the independence of the extrinsic information exchanged in the iterative decoder.
{Moreover, problematic combinatorial objects that cause the error-floor phenomenon, e.g., absorbing sets and trapping sets \cite{Richardson2003,DolecekIT2010}, are formed of cycles with relatively short lengths \cite{EsfahanizadehTCOM2018,HareedyGC2018,BanihashemiIT2014,BanihashemiIT2016}. Finally, short cycles can have a negative impact on the code minimum distance. In \cite{SmarandacheIT2012,BattaglioniTCOM2018}, some upper bounds on the minimum distance of circulant-based  block and SC LDPC codes are derived, and it is shown that the smaller the girth of the graph, the smaller the minimum distance upper bound will be. Thus, improving the girth can result in a larger minimum distance.
}

We present a systematic framework to construct MD-SC codes, which is based on an informed relocation of circulants. MD-SC codes constructed using our proposed framework enjoy a notably lower population of short cycles, and consequently better performance compared to 1D-SC codes. 
Throughout this paper, the operator $\overset{p}{=}$ (resp., $\defopp$\hspace{0.05cm}) defines the congruence (resp., incongruence) modulo $p$, and the operator $(.)_{p}$ defines modulo $p$ of an integer\vspace{0cm}.

\begin{figure}
\centering
\includegraphics[width=0.26\textwidth]{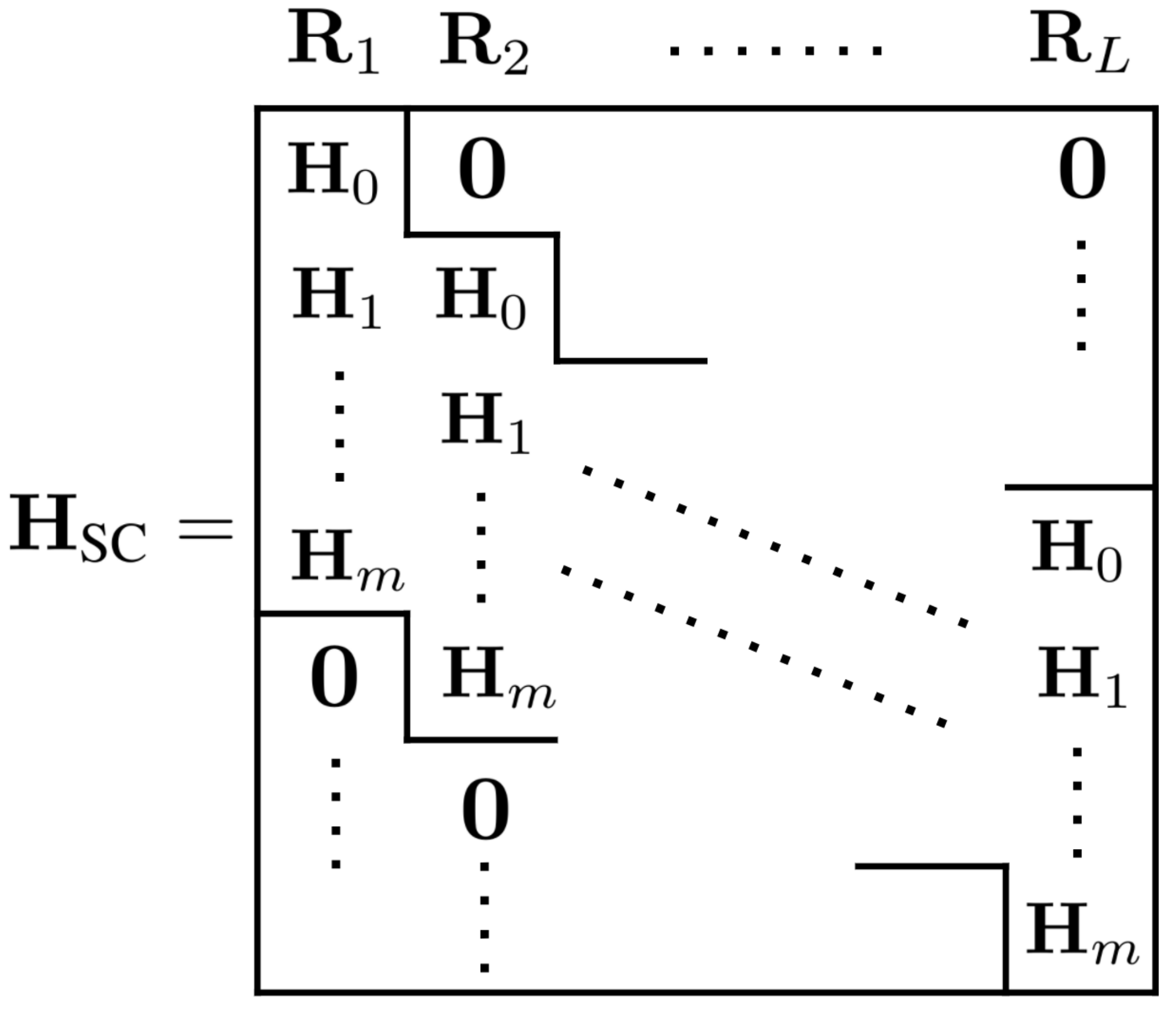}\vspace{0cm}
\caption{The parity-check matrix of an SC code with parameters $m$ and $L$.\vspace{0cm}}
\label{H_SC_structure}
\end{figure}

\newpage
\section{MD-SC Code Structure\vspace{0cm}}
In this section, we demonstrate the structure of our MD-SC codes. Our MD-SC codes have two main parameters: \textit{MD coupling depth} $d$ and \textit{MD coupling length} $L_2$. The parameter $L_2$ of an MD-SC code shows the number of SC codes that are connected together to form the MD-SC code. Each constituent SC code is connected to at most ($d-1$) following SC codes, sequenced in a cyclic order. Thus, $1\leq d\leq L_2$, and $d=1$ corresponds to $L_2$ disjoint 1D-SC codes.

{We intend to reduce the population of cycles with length $k$, or cycles-$k$, in our MD-SC code construction, and the parameter $k$ is an input to our scheme. A wise choice for $k$ is the girth \cite{MacKayIT1999}, or the length of the cycle that is the common denominator of several problematic combinatorial objects for a specific channel, e.g., AWGN channels \cite{EsfahanizadehTCOM2018}, partial response channels \cite{EsfahanizadehTMAG2017}, or Flash channels \cite{Hareedy2017ITW}. For instance, a cycle-$6$ is the common denominator of problematic combinatorial objects for AWGN channels, and a cycle-$8$ is the common denominator of problematic combinatorial objects for 
partial response channels even if the girth is $6$.}

{
An \textit{Auxiliary matrix} $\mathbf{A}_{t}$, $t\in\{1,\cdots,L_2-1\}$, has the same size as the parity-check matrix of the constituent 1D-SC code , i.e., $\mathbf{H}_\textnormal{SC}$, and appears in the parity-check matrix of the final MD-SC code, see (\ref{MD_structure}). The auxiliary matrices are all-zero matrices at the beginning of the framework and are filled with non-zero circulants during the construction process. A \textit{relocation} is defined as moving a non-zero circulant of $\mathbf{H}_\textnormal{SC}$ to the same position in one of the auxiliary matrices.}

Consider an SC code with parity-check matrix $\mathbf{H}_\textnormal{SC}$, memory $m$, and coupling length $L$ as the constituent 1D-SC code, and let  $\mathbf{R}_\nu$ be the middle replica of $\mathbf{H}_\textnormal{SC}$, i.e., $\nu=\lceil L/2 \rceil$. There are $\kappa\gamma$ non-zero circulants in this replica. Out of these $\kappa\gamma$ circulants, we choose $\mathcal{T}$ circulants that are the most problematic, i.e., that contribute to the highest number of cycles-$k$. The parameter $\mathcal{T}$ is called the \textit{MD coupling density}. We relocate the chosen circulants to auxiliary matrices $\mathbf{A}_1$, $\mathbf{A}_2$, $\dots$, $\mathbf{A}_{d-1}$ such that a relocated circulant from $\mathbf{H}_{\textnormal{SC}}$ is moved to the same position in one of the auxiliary matrices. The same relocations are repeated for all the ($L-1$) remaining replicas.
As a result,\vspace{0cm}
\begin{equation}
\mathbf{H}_\textnormal{SC}=\mathbf{H}_\textnormal{SC}'+\sum_{t=1}^{d-1}\mathbf{A}_{t},\vspace{0cm}
\end{equation}
where $\mathbf{H}_\textnormal{SC}'$ is derived from $\mathbf{H}_\textnormal{SC}$ by removing the $\mathcal{T}$ chosen circulants. We note that the middle replica $\mathbf{R}_\nu$ is considered for ranking the circulants in order to include all possible cycles-$k$ that a non-zero circulant in $\mathbf{H}_\textnormal{SC}$ can contribute to.
The parity-check matrix of the MD-SC code, $\mathbf{H}_\textnormal{SC}^\textnormal{MD}$, is constructed as follows, where $\mathbf{A}_{d}=\mathbf{A}_{d+1}=\dots=\mathbf{A}_{L_2-1}=\mathbf{0}$: (The non-zero auxiliary matrices are $\mathbf{A}_1$, $\mathbf{A}_2$, $\dots$, $\mathbf{A}_{d-1}$.)\vspace{0cm}
\begin{equation}
\label{MD_structure}
\mathbf{H}_\textnormal{SC}^\textnormal{MD}=\left[\begin{array}{cccc}
\mathbf{H}_\textnormal{SC}'&\mathbf{A}_{L_2-1}&\cdots&\mathbf{A}_{1}\\
\mathbf{A}_{1}&\mathbf{H}_\textnormal{SC}'&\cdots&\mathbf{A}_{2}\\
\vdots&\vdots&\ddots&\vdots\\
\mathbf{A}_{L_2-1}&\mathbf{A}_{L_2-2}&\cdots&\mathbf{H}_\textnormal{SC}'
\end{array}\right]{.}
\end{equation}

$\mathbf{H}_\textnormal{SC}^\textnormal{MD}$ can be viewed as a collection of $L_2$ rows and $L_2$ columns of \textit{segment}s $\mathbf{S}_{a,b}$, where $0\leq a\leq L_2-1$ and $0\leq b\leq L_2-1$. Each segment $\mathbf{S}_{a,b}$ is a matrix with the same dimension as $\mathbf{H}_\textnormal{SC}$, $\mathbf{S}_{a,a}=\mathbf{H}_\textnormal{SC}'$, $\mathbf{S}_{(a+t)_{\substack{\scriptscriptstyle L_2}},a}=\mathbf{A}_{t}$ for $t\in\{1,\cdots,d-1\}$, and $\mathbf{S}_{(a+t)_{\substack{\scriptscriptstyle L_2}},a}=\mathbf{0}$ for $t\in\{d,\cdots,L_2-1\}$\footnote{{magenta}
{While $\mathbf{H}_\textnormal{SC}^\textnormal{MD}$ may look similar to a block LDPC code, we would like to note that the locality of connections is always preserved during relocations. As such, the code does not reduce to a block LDPC code even when $d=L_2$.}\vspace{0cm}}.\vspace{0cm}

\begin{example}
\label{example_N_4}
Consider an SC code with $\gamma=2$, $\kappa=3$, $z=3$, $m=1$, and $L=3$. The matrix $\mathbf{H}$ of the underlying block code and the component matrices are given below:
\begin{equation*}
\mathbf{H}=\left[\begin{array}{ccc}
\sigma^{f_{0,0}}&\sigma^{f_{0,1}}&\sigma^{f_{0,2}}\\
\sigma^{f_{1,0}}&\sigma^{f_{1,1}}&\sigma^{f_{1,2}}
\end{array}\right]\hspace{-0.1cm},\hspace{0.1cm}
\mathbf{H}_0=\left[\begin{array}{ccc}
\sigma^{f_{0,0}}&\mathbf{0}&\sigma^{f_{0,2}}\\
\mathbf{0}&\sigma^{f_{1,1}}&\mathbf{0}
\end{array}\right]\hspace{-0.1cm},\hspace{0.1cm}
\mathbf{H}_1=\left[\begin{array}{ccc}
\mathbf{0}&\sigma^{f_{0,1}}&\mathbf{0}\\
\sigma^{f_{1,0}}&\mathbf{0}&\sigma^{f_{1,2}}
\end{array}\right]\hspace{-0.1cm}.
\end{equation*}
We intend to construct an MD-SC code with parameters $\mathcal{T}=1$, $d=2$, and $L_2=4$. Assume $\sigma^{f_{1,0}}$ is the most problematic circulant, and we relocate it to $\mathbf{A}_1$. This relocation is applied to all $L=3$ instances of the problematic circulant. We remind that each circulant corresponds to a group of $z$ connections in the graph of the SC code. Four constituent SC codes along with their problematic connections are depicted in Fig.~\ref{MDGraphicalExample}(a). The problematic connections are rewired to the same positions in the next SC codes, in a cyclic order, to construct the MD-SC code, Fig.~\ref{MDGraphicalExample}(b).\vspace{0cm}
\begin{figure}
\centering
\begin{tabular}{cc}
\hspace{-0.47cm}
\includegraphics[width=0.485\textwidth]{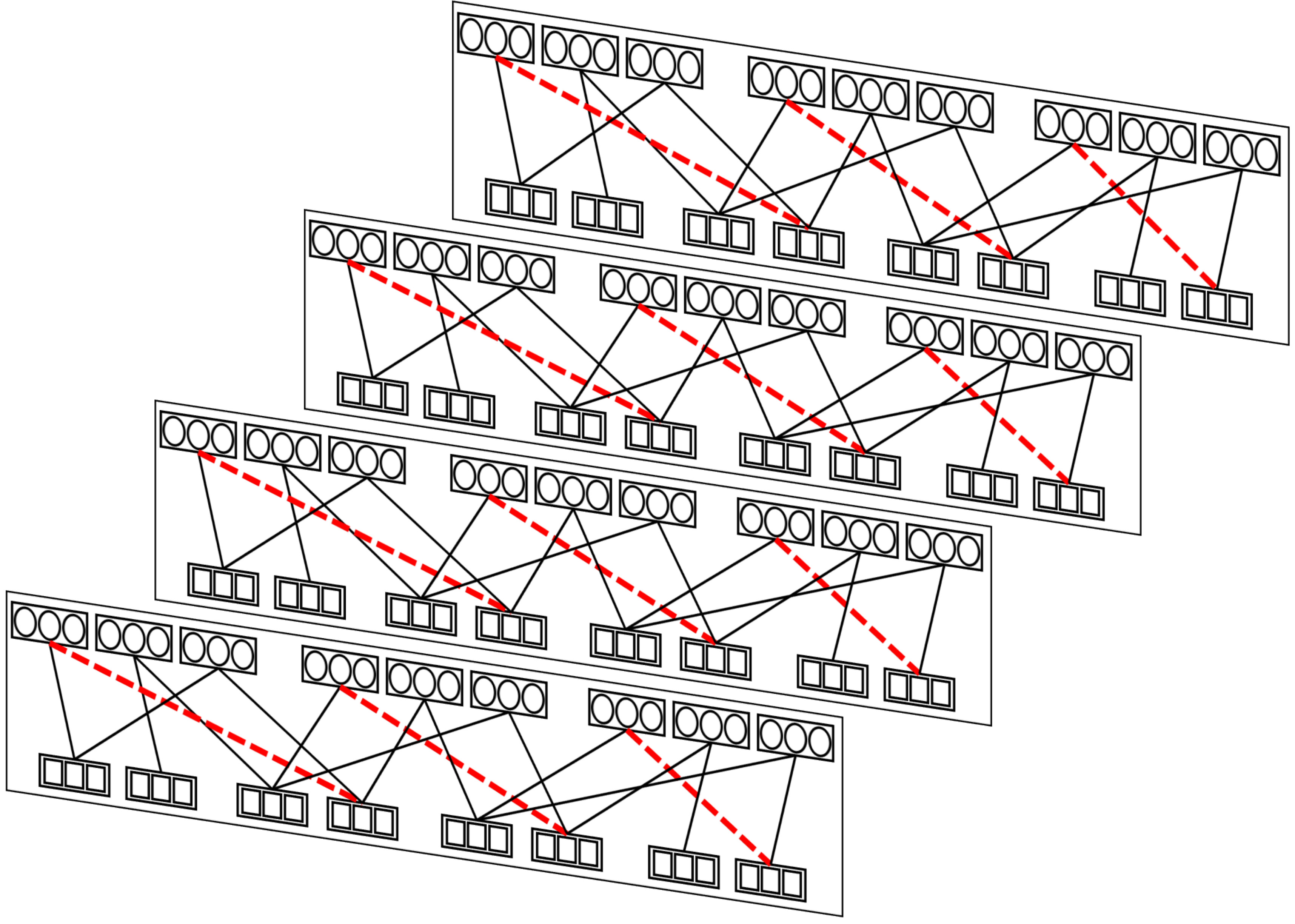}\hspace{-0.7cm}&\hspace{-0.5cm}
\includegraphics[width=0.485\textwidth]{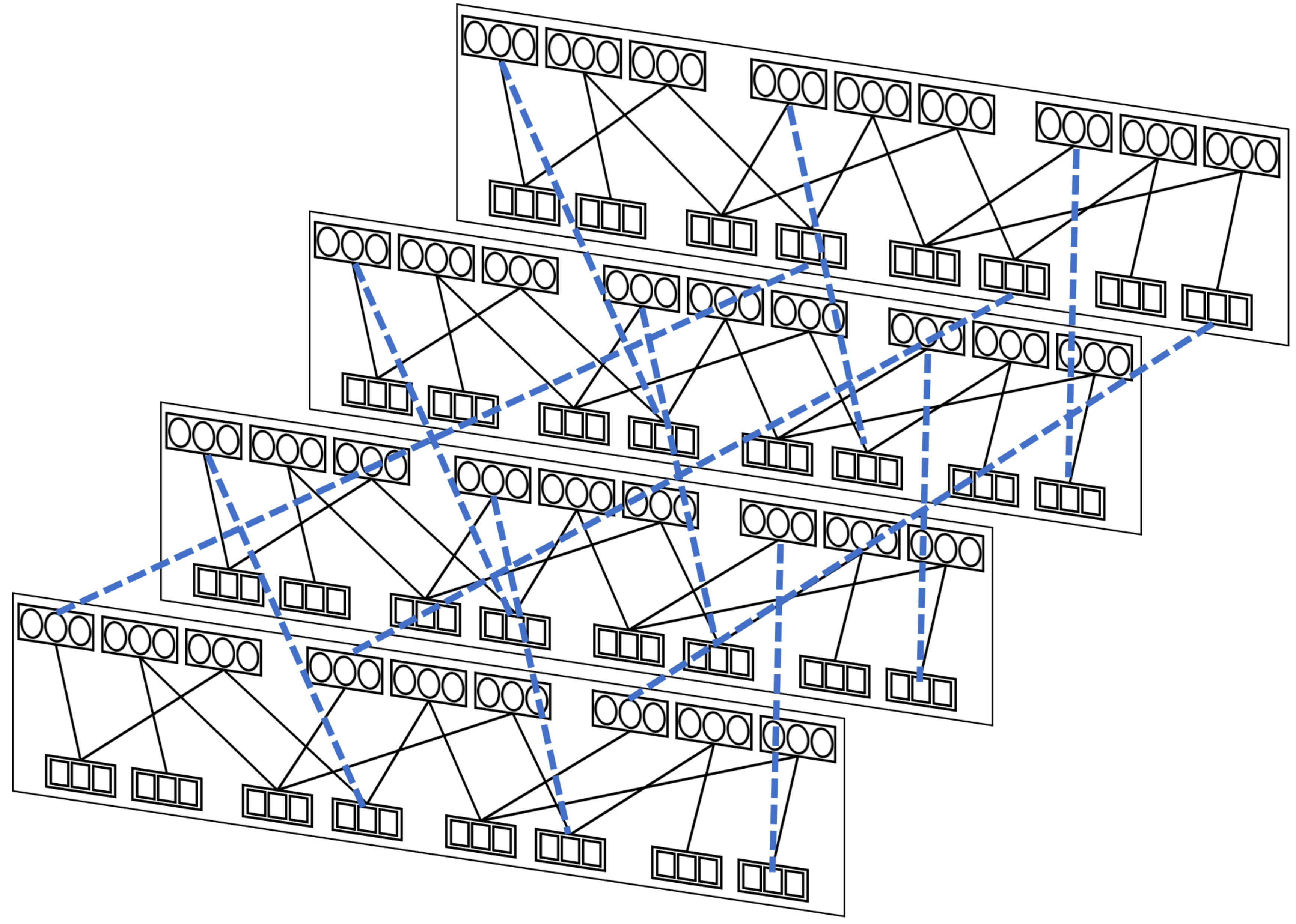}\\
(a)&(b)\vspace{0cm}
\end{tabular}\vspace{0cm}
\caption{(a) Four 1D-SC codes. Circles (resp., squares) represent VNs (resp., CNs). Each line represents a group of connections (defined by a circulant) from $z$ VNs to $z$ CNs. Problematic connections are shown in dashed red lines. (b) MD-SC code with $\mathcal{T}=1$, $d=2$, and $L_2=4$. Rewired connections are shown in dashed blue lines.\vspace{0cm}}
\label{MDGraphicalExample}
\end{figure}
\end{example}


\begin{definition}\textnormal{ }\vspace{0cm}
\begin{enumerate}
\item Let $\mathcal{C}_{i,j}$, where $0\leq i\leq (L+m)\gamma {-}1$ and $0\leq j\leq L\kappa {-} 1$, be a non-zero circulant in $\mathbf{H}_\textnormal{SC}$. We say $\mathcal{C}_{i,j}$ is relocated to $\mathbf{A}_t$, where $t\in\{1,\cdots,d-1\}$, if it is moved from $\mathbf{H}_\textnormal{SC}$ to $\mathbf{A}_t$. We denote this relocation as $\mathcal{C}_{i,j}{\rightarrow}\mathbf{A}_t$.
\item $\mathcal{C}_{i,j}\textnormal{@}\mathbf{S}_{a,b}$ refers to the circulant $\mathcal{C}_{i,j}$ in segment $\mathbf{S}_{a,b}$. When $\mathcal{C}_{i,j}{\rightarrow}\mathbf{A}_t$, the value of $\mathcal{C}_{i,j}\textnormal{@}\mathbf{S}_{a,a}$ is copied to $\mathcal{C}_{i,j}\textnormal{@}\mathbf{S}_{(a+t)_{\substack{\scriptscriptstyle L_2}},a}$, and $\mathcal{C}_{i,j}\textnormal{@}\mathbf{S}_{a,a}$ becomes zero ($a\in\{0,\cdots,L_2-1\}$ and $t\in\{1,\cdots,d-1\}$).
\item The MD mapping $M:\{\mathcal{C}_{i,j}\}{\rightarrow}\{0,\cdots,d-1\}$ is a mapping from a non-zero circulant in $\mathbf{H}_\textnormal{SC}$ to an integer in $\{0,\cdots,d-1\}$, and it is defined as follows:
\begin{enumerate}
\item If $\mathcal{C}_{i,j}{\rightarrow}\mathbf{A}_t$, $M(\mathcal{C}_{i,j})=t$.
\item {If $\mathcal{C}_{i,j}$ is kept in $\mathbf{H}_\textnormal{SC}'$ (no relocation), $M(\mathcal{C}_{i,j})=0$.}
\end{enumerate}
\item A cycle-$k$, or $\mathcal{O}_k$, visits $k$ circulants in the parity-check matrix of the code. We list the $k$ circulants of $\mathcal{O}_k$, according to the order they are visited when the cycle is traversed in a clockwise direction, in a sequence as $C_{\mathcal{O}_k}=\{\mathcal{C}_{i_1,j_1},\mathcal{C}_{i_2,j_2},\dots,\mathcal{C}_{i_k,j_k}\}$, where $i_1=i_2,j_2=j_3,\dots,i_{k-1}=i_k,j_k=j_1$. A circulant can be visited more than once, e.g., Fig.~\ref{cyclesexamples}.
\item We denote the distance between two circulants $\mathcal{C}_{i_u,j_u}$ and $\mathcal{C}_{i_v,j_v}$ on a cycle $\mathcal{O}_k$, where $u,v\in\{1,\dots,k\}$, as $D_{\mathcal{O}_k}(\mathcal{C}_{i_u,j_u},\mathcal{C}_{i_v,j_v})\in\{0,\dots,k-1\}$.
By definition, $D_{\mathcal{O}_k}(\mathcal{C}_{i_u,j_u},\mathcal{C}_{i_v,j_v})=|v-u|$.
\end{enumerate}\vspace{0cm}
\end{definition}
\begin{figure}
\centering
\vspace{0cm}
\begin{tabular}{cc}
\includegraphics[width=0.25\textwidth]{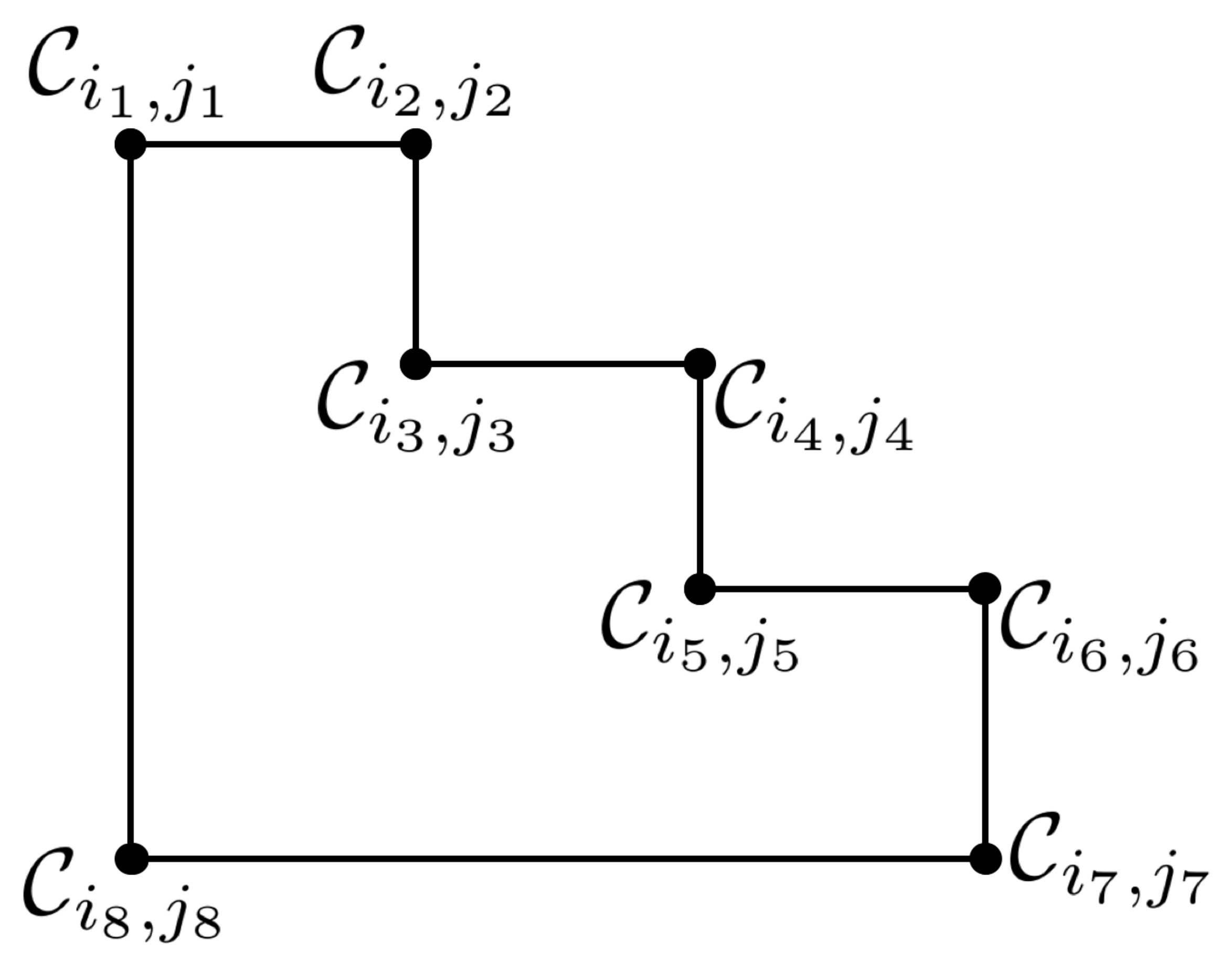}&
\includegraphics[width=0.25\textwidth]{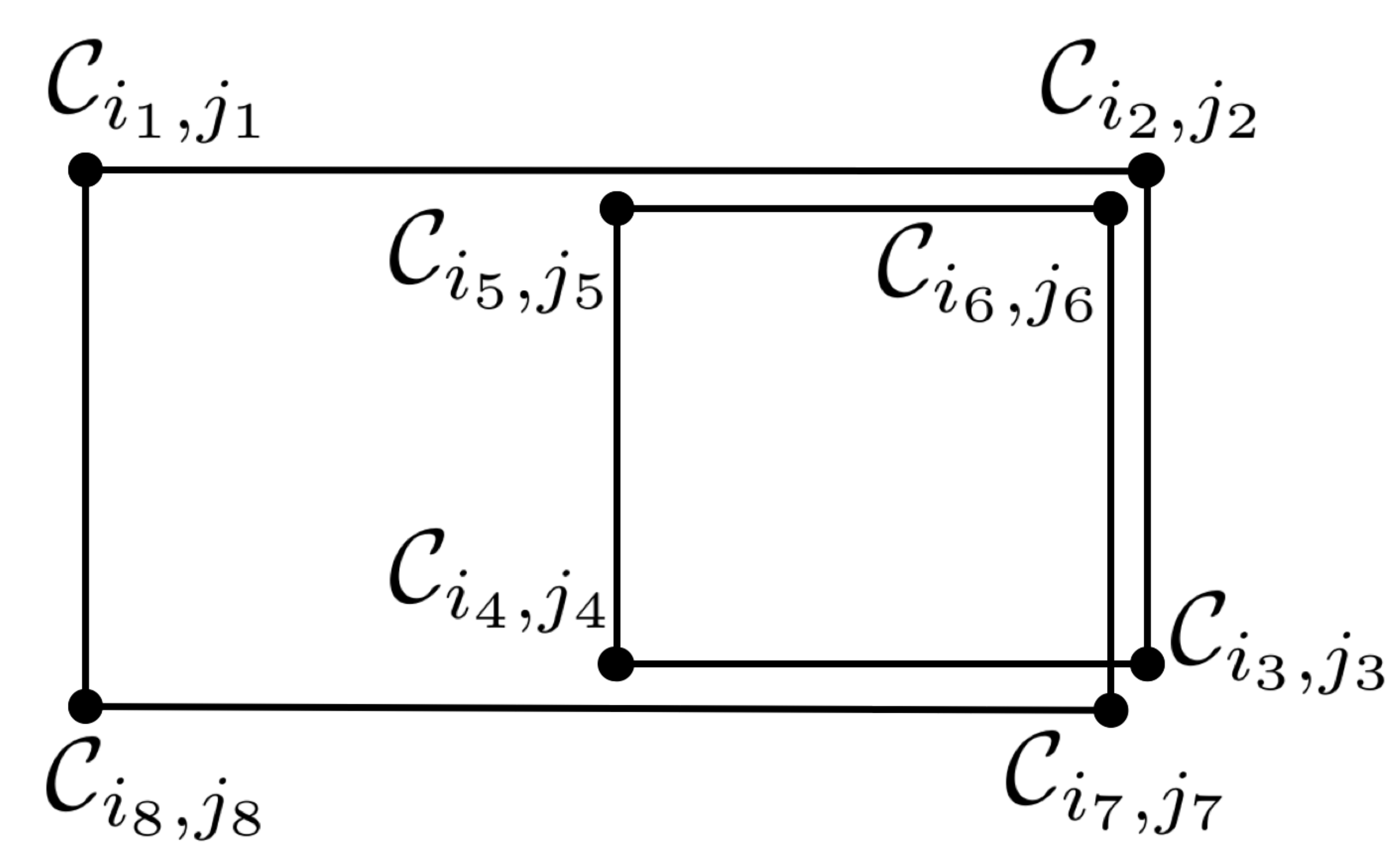}\vspace{0cm}\\
(a)&(b)\vspace{0cm}
\end{tabular}\vspace{0cm}
\caption{Cycles-$8$ with $C_{\mathcal{O}_8}=\{\mathcal{C}_{i_1,j_1},\dots,\mathcal{C}_{i_8,j_8}\}$. Each line represents a connection between two circulants. (a) All circulants are unique. (b) $\mathcal{C}_{i_6,j_6}=\mathcal{C}_{i_2,j_2}$ and $\mathcal{C}_{i_7,j_7}=\mathcal{C}_{i_3,j_3}$.\vspace{0cm}}
\label{cyclesexamples}
\end{figure}
In the new MD-SC code design framework, we effectively answer two questions: which circulants to relocate, and where to relocate them.\vspace{0cm} 

\section{Novel Framework for MD-SC Code Design\vspace{0cm}}
In this section, we present a new framework for constructing MD-SC codes. First, we investigate the effects of relocating a subset of circulants on the population of cycles. Then, we present our algorithm for constructing MD-SC codes which is based on a score voting policy.\vspace{0cm}

\subsection{The Effects of Relocation of Circulants on Cycles}
Consider a cycle $\mathcal{O}_k$ in $\mathbf{H}_\textnormal{SC}$ with sequence of circulants $C_{\mathcal{O}_k}$. Prior to any relocation, there are $L_2$ instances of this cycle in the MD-SC code with parameter $L_2$, one per each constituent SC code.
We investigate the effect of relocating a subset of circulants of $\mathcal{O}_k$, and we call this subset \textit{targeted circulants}. We show that, after relocations, $L_2$ instances of circulants of $C_{\mathcal{O}_k}$ can form $L_2$ cycles of length $k$, $L_2/2$ cycles of length $2k$, $\dots$, or one cycle of length $L_2k$. The first case is a result of bad choices for relocations, and the rest are more preferable. In fact, we opt for the relocations that result in larger cycles (with smaller cardinality as a result).\vspace{0cm}

\begin{theorem}
Let $C_{\mathcal{O}_k}=\{\mathcal{C}_{i_1,j_1},\mathcal{C}_{i_2,j_2},\dots,\mathcal{C}_{i_k,j_k}\}$ be the sequence of circulants in $\mathbf{H}_\textnormal{SC}$ that are visited in a clockwise order by $\mathcal{O}_k$. If the following equation holds, the $L_2$ instances of circulants of $C_{\mathcal{O}_k}$ form $L_2$ cycles-$k$ in $\mathbf{H}_\textnormal{SC}^\textnormal{MD}$,\vspace{0cm}
\begin{equation}\label{mainIneffReloc}
\sum_{u=1}^{k}(-1)^{u}M(\mathcal{C}_{i_u,j_u})\overset{L_2}{=}0.
\vspace{0cm}
\end{equation}
Otherwise, the instances of the targeted circulants do not result in cycles-$k$ in $\mathbf{H}_\textnormal{SC}^\textnormal{MD}$
\footnote{{Equation (\ref{mainIneffReloc}) resembles Fossorier's condition on circulant powers of a CB code that makes a cycle in the protograph result in multiple cycles in the lifted graph of the code \cite{FossorierIT2004}.}}. We call (\ref{mainIneffReloc}) the \textit{Ineffective Relocation Condition}, or IRC, in the rest of this paper.\vspace{0cm}
\end{theorem}
\begin{proof}
Let $(\mathcal{C}_{i_u,j_u},\mathcal{C}_{i_{u+1},j_{u+1}})$ be a pair of consecutive circulants in $C_{\mathcal{O}_k}$, where $u\in\{1,\dots,k\}$ and $\mathcal{C}_{i_{k+1},j_{k+1}}=\mathcal{C}_{i_1,j_1}$. By definition, two circulants have the same row (resp., column) group index, i.e., $i_u=i_{u+1}$ (resp., $j_u=j_{u+1}$), when $u\overset{2}{=}1$ (resp., $u\overset{2}{=}0$).
Before relocations, $\mathcal{C}_{i_u,j_u}\textnormal{@}\mathbf{S}_{a,a}\neq\mathbf{0}$ and $\mathcal{C}_{i_u,j_u}\textnormal{@}\mathbf{S}_{a,b}=\mathbf{0}$, where $\mathcal{C}_{i_u,j_u}\in C_{\mathcal{O}_k}$, $a,b\in\{0,\cdots,L_2-1\}$, and $a\neq b$. This results in $L_2$ instances of $\mathcal{O}_k$, one per each segment $\mathbf{S}_{a,a}$. After relocations, the circulants in $C_{\mathcal{O}_k}$ do not all belong to the same segment. 

Here, a unit of a MD horizontal (resp., MD vertical) shift is defined as cyclically going one segment right (resp., down) when we go from $\mathcal{C}_{i_u,j_u}$ to $\mathcal{C}_{i_{u+1},j_{u+1}}$.
The cycle $\mathcal{O}_k$ reflects in the MD-SC code as cycles with the same length $k$ if and only if when we start from $\mathcal{C}_{i_1,j_1}\neq\mathbf{0}$ from one segment and traverse the circulants of the cycle  in a clock wise order (with the same order they appear in $C_{\mathcal{O}_k}$), we end up at the same segment that we started with.

The segments of $\mathbf{H}_\textnormal{SC}^\textnormal{MD}$ appear in the cyclic order $\{\mathbf{H}_\textnormal{SC}',\mathbf{A}_{L_2-1},\cdots,\mathbf{A}_1\}$, with the MD mapping $\{0,L_2-1,\cdots,1\}$, from left to right. These segments appear in the cyclic order $\{\mathbf{H}_\textnormal{SC}',\mathbf{A}_1,\cdots,\mathbf{A}_{L_2-1}\}$, with the MD mapping $\{0,1,\cdots,L_2-1\}$, from top to bottom, see (\ref{MD_structure}). Thus, the MD horizontal shift, when we go from $\mathcal{C}_{i_u,j_u}$ to $\mathcal{C}_{i_{u+1},j_{u+1}}$, $u\in\{1,3,\dots,k-1\}$, is $(M(\mathcal{C}_{i_u,j_u})-M(\mathcal{C}_{i_{u+1},j_{u+1}}))_{\substack{{ }\\{L_2}}}$ units, see Definition 1.3. Similarly, the MD vertical shift, when we go from $\mathcal{C}_{i_u,j_u}$ to $\mathcal{C}_{i_{u+1},j_{u+1}}$, $u\in\{0,2,\dots,k\}$, is $(M(\mathcal{C}_{i_{u+1},j_{u+1}})-M(\mathcal{C}_{i_u,j_u}))_{\substack{{ }\\{L_2}}}$ units. We remind that the operator $(.)_{p}$ defines modulo $p$ of an integer. 
The total MD horizontal and vertical shifts when we traverse the circulants of $\mathcal{O}_k$ in $\mathbf{H}_\textnormal{SC}^\textnormal{MD}$ are $\delta_{H}$ and $\delta_{V}$, respectively:\vspace{0cm}
\begin{equation}
\begin{split}
&\delta_{H}=(\hspace{-0.5cm}\sum_{u\in\{1,3\dots,k-1\}}\hspace{-0.5cm}[M(\mathcal{C}_{i_u,j_u})-M(\mathcal{C}_{i_{u+1},j_{u+1}})])_{\substack{{ }\\{L_2}}}=\hspace{-0.07cm}(-\sum_{u=1}^{k}[(-1)^{u}M(\mathcal{C}_{i_u,j_u})])_{\substack{{ }\\{L_2}}}{,}\\
&\delta_{V}=(\hspace{-0.34cm}\sum_{u\in\{2,4\dots,k\}}\hspace{-0.34cm}[M(\mathcal{C}_{i_{u+1},j_{u+1}})-M(\mathcal{C}_{i_u,j_u})])_{\substack{{ }\\{L_2}}}=\hspace{-0.07cm}(-\sum_{u=1}^{k}[(-1)^{u}M(\mathcal{C}_{i_u,j_u})])_{\substack{{ }\\{L_2}}}.\vspace{0cm}
\end{split}
\vspace{-0cm}
\end{equation}

The relocations are ineffective if and only if the start and end segments are the same when we traverse the $k$ circulants of $\mathcal{O}_k$. For this to happen, the total MD horizontal and vertical shifts ($\delta_{H}$ and $\delta_{V}$) need to be zero, which results in (\ref{mainIneffReloc}).
\end{proof}


If equation (\ref{mainIneffReloc}), or IRC, holds for the circulants of $\mathcal{O}_k$, $L_2$ instances of circulants of $C_{\mathcal{O}_k}$ in $\mathbf{H}_\textnormal{SC}^\textnormal{MD}$ form $L_2$ cycles-$k$ in the MD-SC code (unpreferable). Theorem~2 investigates the situation when IRC does not necessarily hold.\vspace{0cm}

\begin{theorem}
Each cycle $\mathcal{O}_k$ in $\mathbf{H}_\textnormal{SC}$ results in $\tau$ cycles with length $L_2k/\tau$ in $\mathbf{H}_\textnormal{SC}^\textnormal{MD}$, where
\begin{equation}
\tau=\gcd(L_2,\Delta_{\mathcal{O}_k}),
\text{ and }\hspace{1cm}
\Delta_{\mathcal{O}_k}=(-\sum_{u=1}^{k}[(-1)^{u}M(\mathcal{C}_{i_u,j_u})])_{\substack{\\\scriptstyle{L_2}}}.\vspace{0cm}
\label{Theorem2Equ}
\end{equation}
The operator $\gcd$ outputs the greatest common divisor of its two operands.\vspace{0cm}
\end{theorem}
\begin{proof}
Consider a cycle $\mathcal{O}_k$ with $C_{\mathcal{O}_k}=\{\mathcal{C}_{i_1,j_1},\dots,\mathcal{C}_{i_k,j_k}\}$ in $\mathbf{H}_\textnormal{SC}$. There are $(L_2)^2$ instances of $\mathcal{C}_{i_u,j_u}$ in $\mathbf{H}_\textnormal{SC}^\textnormal{MD}$, $u\in\{1,\dots,k\}$, one per each segment, and only $L_2$ of them can be non-zero. We traverse the circulants of $\mathcal{O}_k$ in $\mathbf{H}_\textnormal{SC}^\textnormal{MD}$ according to the order they appear in $C_{\mathcal{O}_k}$ starting from a non-zero instance of $\mathcal{C}_{i_1,j_1}$.
After traversing all $k$ circulants, we reach circulant $\mathcal{C}_{i_1,j_1}$ in a segment that is (cyclically) $\Delta_{\mathcal{O}_k}$ units right and $\Delta_{\mathcal{O}_k}$ units down from the segment we started.

If $\Delta_{\mathcal{O}_k}=0$, the cycle is complete after traversing the $k$ circulants. In this case, there are $L_2$ instances of $C_{\mathcal{O}_k}$, one per each non-zero instance of $\mathcal{C}_{i_1,j_1}$. If $\Delta_{\mathcal{O}_k}\neq0$, the cycle cannot be complete after traversing $k$ circulants. We proceed traversing the circulants until we reach $\mathcal{C}_{i_1,j_1}$ that is in the same segment that we started from. 

We define the parameter $\lambda$ as follows:\vspace{0cm}
\begin{equation}
\lambda=\min\{g|g\in\{1,2,\cdots\},g\Delta_{\mathcal{O}_k}\overset{L_2}{=}0\}.\vspace{0cm}
\end{equation}
Then, we complete the cycle after traversing $\lambda k$ circulants. The parameter $\lambda$ is the minimum integer value such that $\lambda\Delta_{\mathcal{O}_k}\overset{L_2}{=}0$, i.e., $\lambda=L_2/\gcd{(L_2,\Delta_{\mathcal{O}_k})}$. The $L_2$ non-zero instances of the $k$ circulants  in $C_{\mathcal{O}_k}$ form $ \tau = L_2k/\lambda k= \gcd(L_2,\Delta_{\mathcal{O}_k})$ cycles of the length $\lambda k = L_2k/\tau$.\vspace{-0cm}
\end{proof}

For example, when $L_2$ and $\Delta_{\mathcal{O}_k}$ are relatively prime, there is a cycle with length $L_2k$ that traverses all non-zero instances of the circulants of $C_{\mathcal{O}_k}$. When $\tau=\gcd{(L_2,\Delta_{\mathcal{O}_k})}=L_2$, the non-zero instances of the circulants of $C_{\mathcal{O}_k}$ form $L_2$ cycles with length $k$.
In our algorithm for the MD-SC code construction, the relocations that result in smaller $\tau$ are more preferred as they result in larger cycles.\vspace{0cm}
\begin{remark}
Review some properties of $\gcd$ that are used in the rest of this paper:
\begin{itemize}
\item $\gcd(a,0)=|a|$ for any non-zero $a$,
\item $\gcd(a+yb,b)=\gcd(a,b)$ for any integer $y$,
\item $\gcd(-a,b)=\gcd(a,b)$.\vspace{0cm}
\end{itemize}
\end{remark}

\begin{example} Let $C_{\mathcal{O}_k}=\{\mathcal{C}_{i_1,j_1},\dots,\mathcal{C}_{i_k,j_k}\}$ be the sequence of circulants of $\mathcal{O}_k$, and $n$ be the number of its relocated circulants.

\begin{enumerate}[leftmargin=0.5cm]
\item Let $n=1$, $\mathcal{C}_{i_a,j_a}\rightarrow\mathbf{A}_1$, and $L_2=3$. Then, $\Delta_{\mathcal{O}_k}=((-1)^a)_3$ and $\tau=1$. Fig.~\ref{relocation_examples}(a) shows $\mathcal{C}_{i_a,j_a}{\rightarrow}{\mathbf{A}_1}$. Fig.~\ref{relocation_examples}(b) shows that a cycle-$3k$ (shown in orange) is formed. The green  border represents that this relocation is preferable.

\item Let $n=2$, $\mathcal{C}_{i_a,j_a},\mathcal{C}_{i_b,j_b}{\rightarrow}{\mathbf{A}_1}$, and $L_2=3$. Suppose $D_{\mathcal{O}_k}(\mathcal{C}_{i_a,j_a},\mathcal{C}_{i_b,j_b})=1$. Then,  $\Delta_{\mathcal{O}_k}=((-1)^a-(-1)^a)_3=0$ and $\tau=L_2=3$. Fig.~\ref{relocation_examples}(c) shows $\mathcal{C}_{i_a,j_a},\mathcal{C}_{i_b,j_b}{\rightarrow}{\mathbf{A}_1}$. Fig.~\ref{relocation_examples}(d) shows that three cycles-$k$ are formed. The red border represents that these relocations are unpreferable.

\item Let $n=3$, and $\mathcal{C}_{i_a,j_a}$, $\mathcal{C}_{i_b,j_b}$, $\mathcal{C}_{i_c,j_c}{\rightarrow}{\mathbf{A}_2}$, and $L_2=4$. Suppose these three circulants are consecutive in $C_{\mathcal{O}_k}$. Then, $\Delta_{\mathcal{O}_k}=((-1)^a(2-2+2))_{\substack{\\\scriptstyle{4}}}=2$ and $\tau=2$. Fig.~\ref{relocation_examples}(e) shows $\mathcal{C}_{i_a,j_a},\mathcal{C}_{i_b,j_b},\mathcal{C}_{i_c,j_c}{\rightarrow}\mathbf{A}_2$. Fig.~\ref{relocation_examples}(f) shows that two cycles-$2k$ are formed. The red border represents that these relocations are less preferred.  We note that if we relocated the targeted circulants to $\mathbf{A}_1$ instead, the result would be one cycle-$4k$ which is more preferred.\vspace{0cm} 


\end{enumerate}
\begin{figure}\vspace{0cm}
\centering
\begin{tabular}{ccc}
\includegraphics[height=0.16\textwidth]{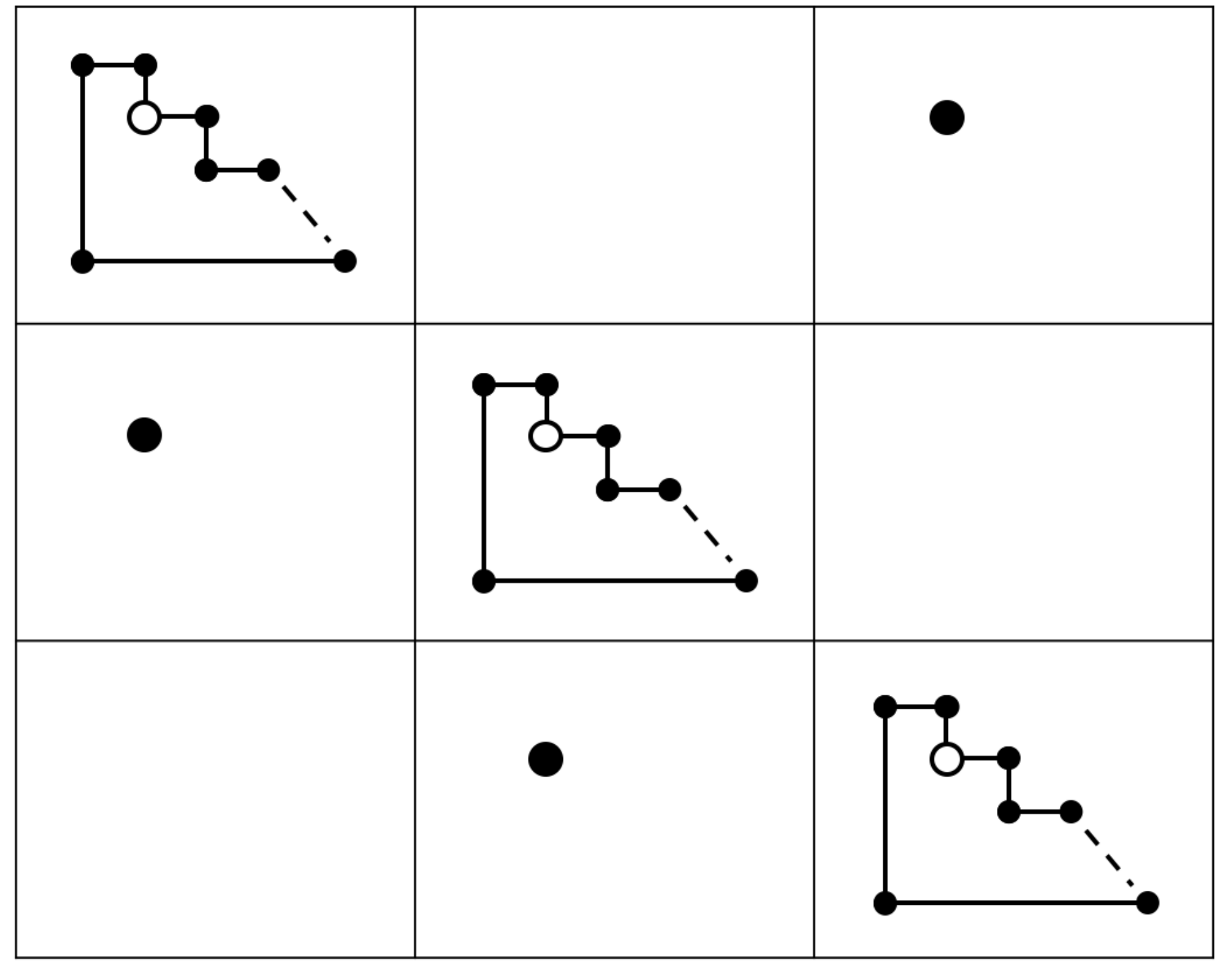}&\includegraphics[height=0.16\textwidth]{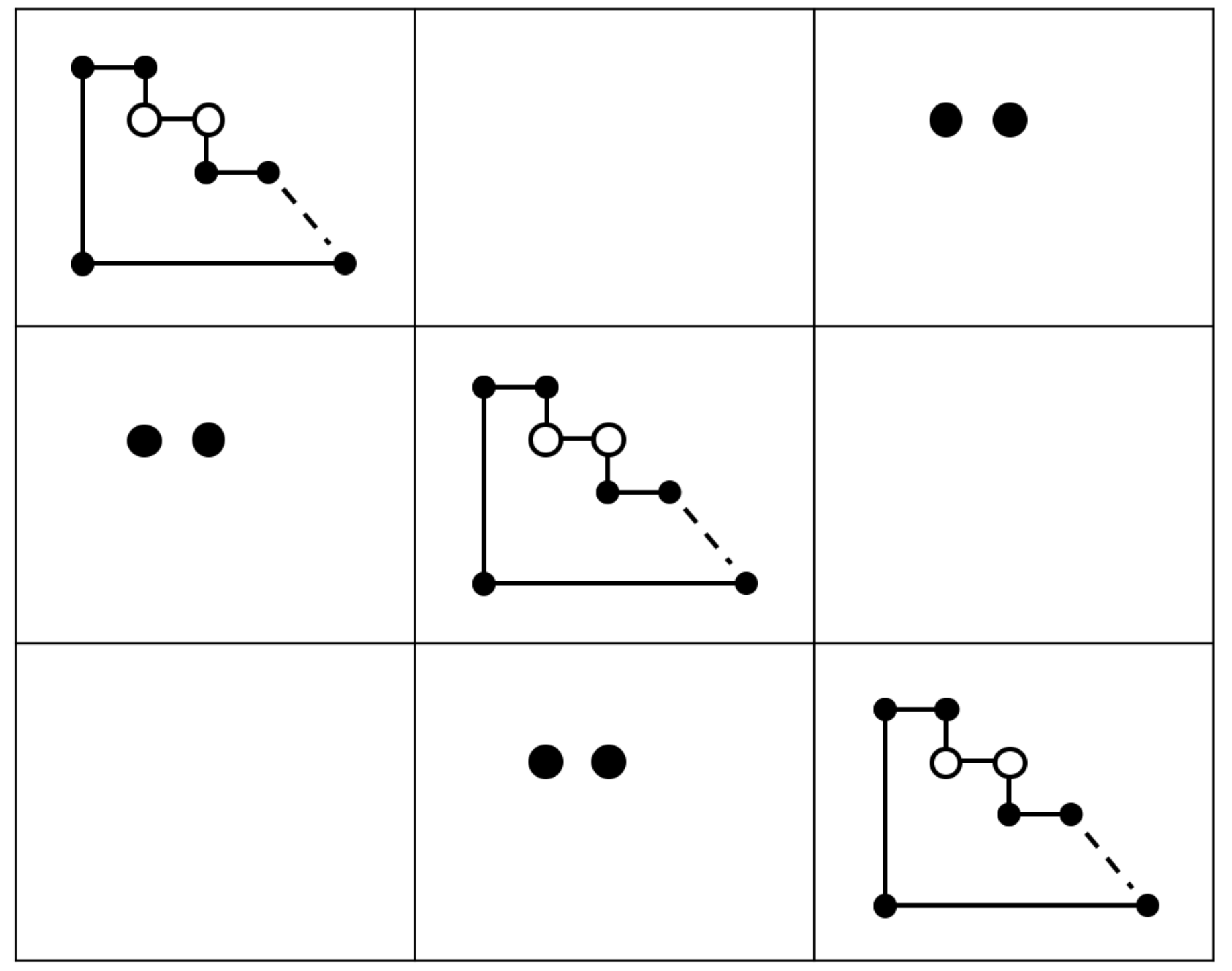}&\includegraphics[height=0.16\textwidth]{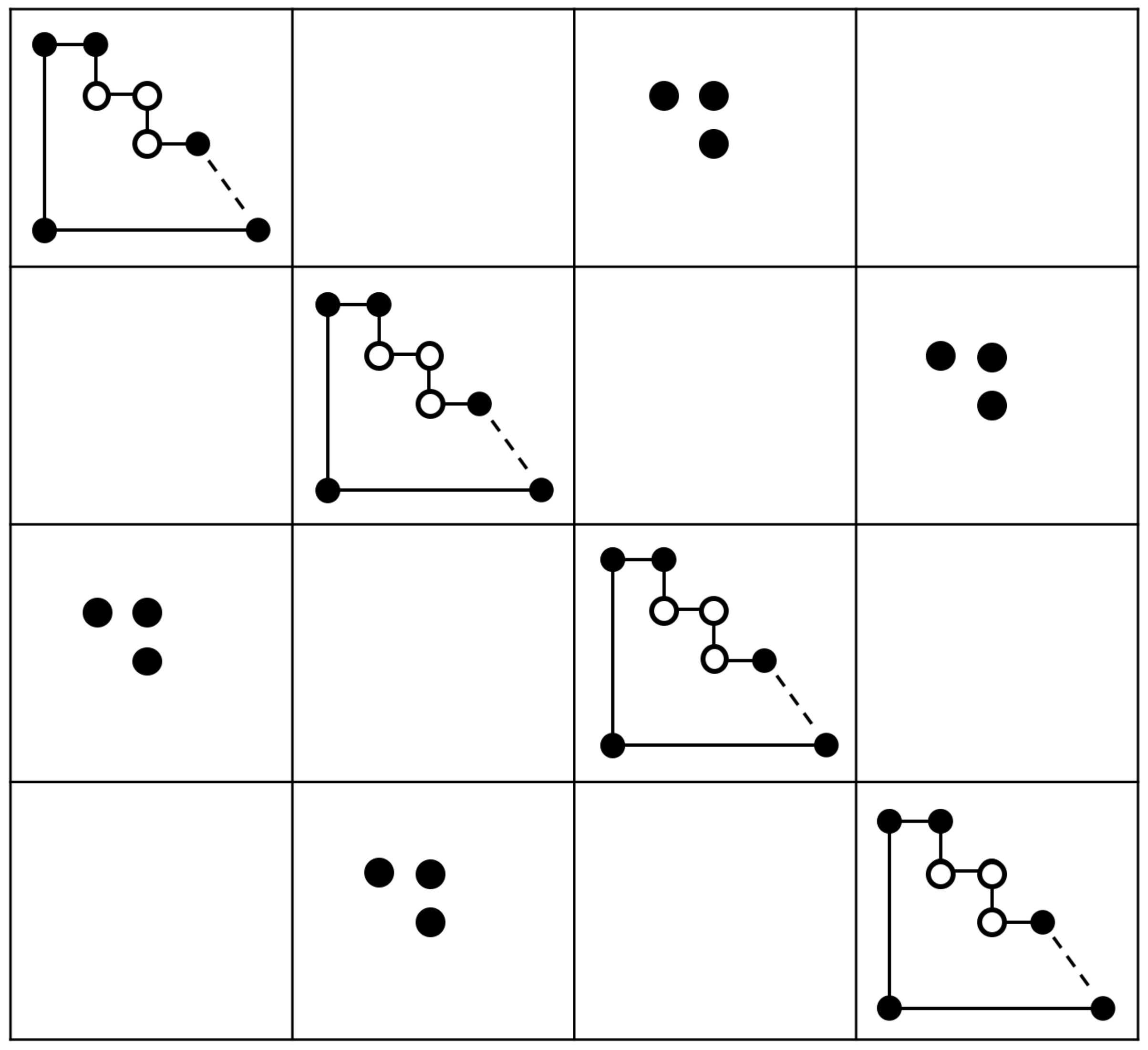}\vspace{0cm}\\
(a)&(c)&(e)\vspace{0cm}\\
\includegraphics[height=0.16\textwidth]{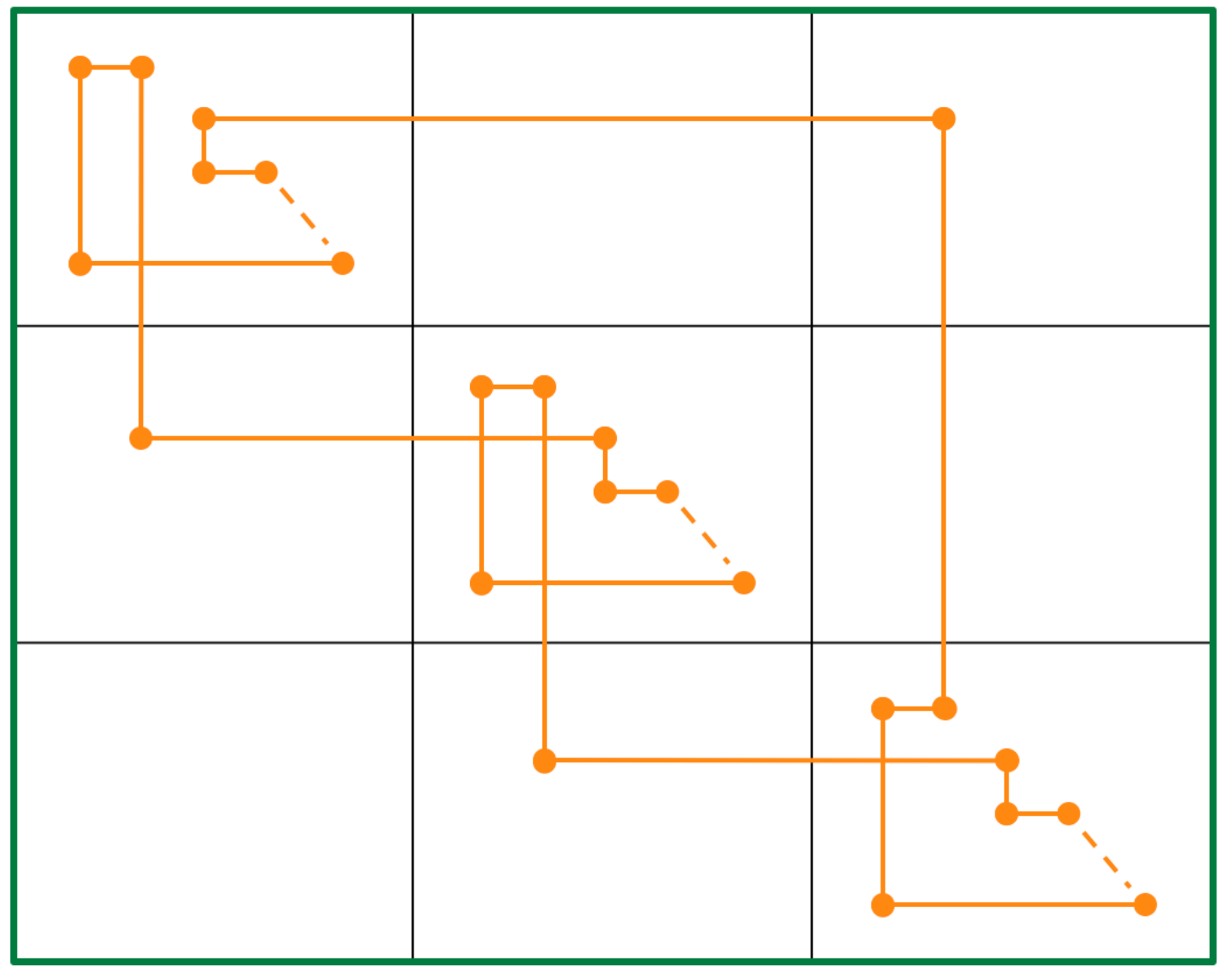}&\includegraphics[height=0.16\textwidth]{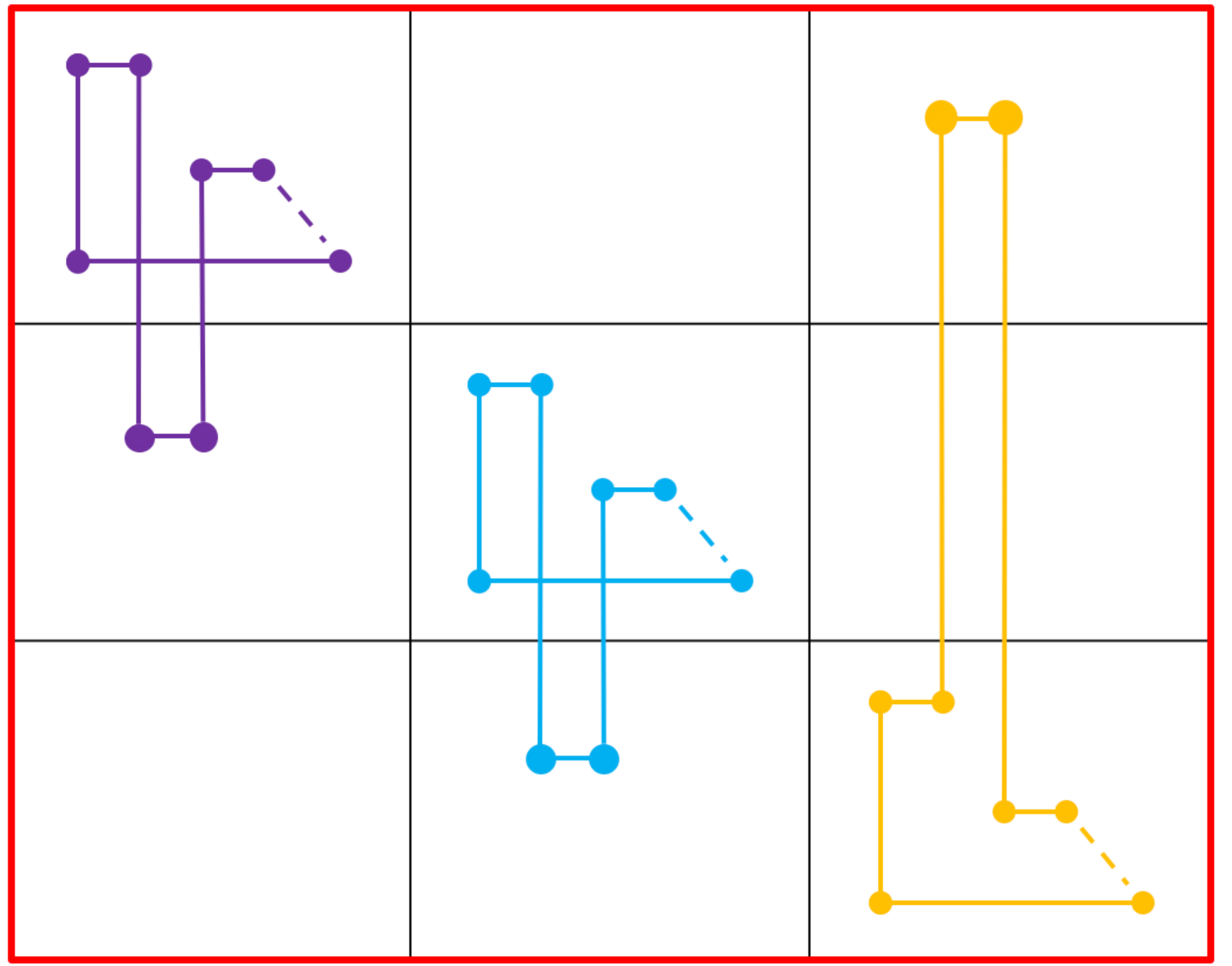}&\includegraphics[height=0.16\textwidth]{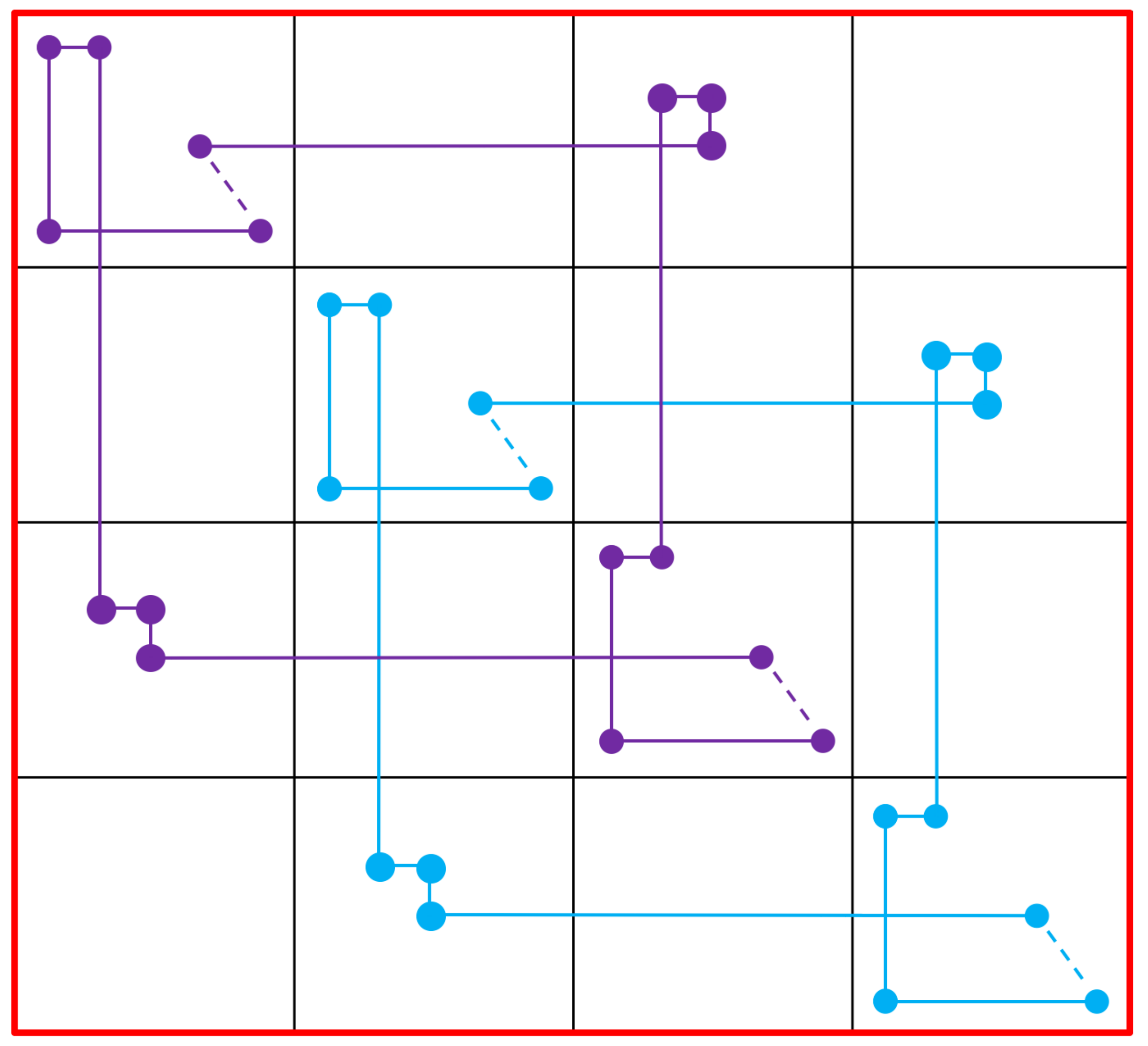}\vspace{-0cm}\\
(b)&(d)&(f)\vspace{-0cm}
\end{tabular}
\caption{(a) $\mathcal{C}_{i_a,j_a}{\rightarrow}\mathbf{A}_1$. The white circles show original locations of the relocated circulant. (b) A cycle-$3k$ is formed. (c) $\{\mathcal{C}_{i_a,j_a},\mathcal{C}_{i_b,j_b}\}{\rightarrow}\mathbf{A}_1$. (d) Three cycles-$k$ are formed. (e) $\{\mathcal{C}_{i_a,j_a},\mathcal{C}_{i_b,j_b},\mathcal{C}_{i_c,j_c}\}{\rightarrow}\mathbf{A}_2$. (f) Two cycles-$2k$ are formed.\vspace{0cm}}
\label{relocation_examples}
\end{figure}
\end{example}

\vspace{-0cm}
\begin{remark}
A circulant can appear more than once in $C_{\mathcal{O}_k}$, e.g., see Fig.~\ref{cyclesexamples}(b). A circulant that is repeated $r$ times in the sequence can be interpreted in our analysis as $r$ different circulants; every  two circulants from this group have an even distance on $\mathcal{O}_k$. The relocation of a circulant that appears $r$ times {is equivalent to the relocation of $r$ circulants 
to the same auxiliary matrix.\vspace{-0cm}}
\end{remark}

\subsection{Score Voting Algorithm for MD-SC Code Design\vspace{-0cm}}
Our framework is based on a score voting policy and aims at minimizing the population of short cycles. As stated in Section~III, the MD coupling with depth $d$ is performed via relocating problematic circulants to auxiliary matrices $\mathbf{A}_t$, $t\in\{1,\cdots,d-1\}$. After relocating one circulant, the ranking of the problematic circulants (with respect to the number of cycles each of them is visited by) changes. Thus, the relocations are performed sequentially. {In our framework, we  use a tree-based strategy for constructing MD-SC codes, by identifying a proper sequence of relocations such that as many as possible designated cycles are removed in the constituent SC codes, while as few as possible short cycles are formed in the multi-dimensional configuration. To assign scores to the branches of the tree, we use the results of Section IV~(A). A tree-based strategy has also been recently applied to find a good partitioning to construct 1D-SC codes with a reduced population of problematic objects \cite{BattaglioniArxiv2019}. 

Consider a targeted circulant $\mathcal{C}_{i_v,j_v}$. There are $d$ possible relocation  options for this circulant: relocate to one of the $(d-1)$ auxiliary matrices or keep in $\mathbf{H}_\textnormal{SC}'$, i.e., $M(\mathcal{C}_{i_v,j_v})=t$ and $t\in\{0,1,\cdots d-1\}$. Each cycle $\mathcal{O}_k$ in $\mathbf{H}_\textnormal{SC}$ that has the targeted circulant in its sequence gives a score for each of these options, and the collective scoring results are considered for making a decision. The score $R(\mathcal{O}_k,t)$ is proportional to the length of the cycles that the non-zero instances of the circulants of $C_{\mathcal{O}_k}$ form after applying the corresponding option (after performing a relocation or keeping the targeted circulant in $\mathbf{H}_\textnormal{SC}'$):\vspace{-0cm}
\begin{equation}
\label{score_equ}
R(\mathcal{O}_k,t)=\frac{L_2}{\gcd(L_2,\Delta_{\mathcal{O}_k})},
\hspace{0.6cm}
\Delta_{\mathcal{O}_k}=((-1)^{v+1}rt-\sum_{\mathcal{C}_{i_u,j_u}\in C_{\mathcal{O}_k}\setminus \mathcal{C}_{i_v,j_v}}[(-1)^{u}M(\mathcal{C}_{i_u,j_u})])_{\substack{\scriptstyle{L_2}}}.\vspace{-0.0cm}
\end{equation}
Here, we assumed $\mathcal{C}_{i_v,j_v}$ is repeated $r$ times in $C_{\mathcal{O}_k}$, and $v$ is the index of one of the repetitions.

In fact, there might be several options for a targeted circulant such that IRC (i.e., (\ref{mainIneffReloc})) does not hold. However, the options that result in larger cycles (with smaller cardinality as a consequence) are preferable. We use a scoring system in our algorithm for constructing MD-SC codes in order to convert short cycles in the constituent SC codes into cycles with lengths as large as possible.\vspace{-0cm}

\begin{example}
Consider the cycle $\mathcal{O}_k$ and a targeted circulant $\mathcal{C}_{i_v,j_v}\in C_{\mathcal{O}_k}$.\\
\underline{Scenario 1:} No circulants of $\mathcal{O}_k$ are previously relocated, and $\mathcal{C}_{i_v,j_v}$ appears once in $C_{\mathcal{O}_k}$ (i.e., $r=1$). Thus, IRC does not hold after a relocation, regardless of the auxiliary matrix that $\mathcal{C}_{i_v,j_v}$ is relocated to. For the option $M(\mathcal{C}_{i_v,j_v})=t$, 
$R(\mathcal{O}_k,t)=L_2/\gcd(L_2,t)$.
For instance, $\mathcal{O}_k$ gives score $1$ to the option ``keep in $\mathbf{H}_\textnormal{SC}'$'', and gives score $L_2$ to the option ``relocate to $\mathbf{A}_1$''.\\
\underline{Scenario 2:} Circulant $\mathcal{C}_{i_w,j_w}\in C_{\mathcal{O}_k}$ is already relocated to $\mathbf{A}_1$, $D_{\mathcal{O}_k}(\mathcal{C}_{i_v,j_v},\mathcal{C}_{i_w,j_w})=2$, and both circulants appear once in $C_{\mathcal{O}_k}$ (i.e., $r=1$). Then, IRC does not hold for options ``no relocation" and ``relocation to $\mathbf{A}_t$", when $t\neq{L_2-1}$. In fact, for the option $M(\mathcal{C}_{i_u,j_u})=t$, $t\in\{0,\cdots,d-1\}$, $\mathcal{O}_k$ gives score 
$R(\mathcal{O}_k,t)=L_2/\gcd(L_2,t+1).$
For instance, $\mathcal{O}_k$ gives score $1$ {blue}{to} ``relocate to $\mathbf{A}_{L_2-1}$'', and gives score $L_2$ to ``keep in $\mathbf{H}_\textnormal{SC}'$'' and ``relocate to $\mathbf{A}_{t'}$'' where  $(t'+1)$ and $L_2$ are relatively prime.\vspace{-0cm}
\end{example}

The relocation options are $\{$relocate to $\mathbf{A}_1$,\dots, relocate to $\mathbf{A}_{d-1}$, keep in $\mathbf{H}_\textnormal{SC}'\}$. We identify the best options for a targeted circulant as follows:
We first identify and keep the options that receive the least number of scores with value $x=1$, as these options result in fewer cycle-$k$ in the MD-SC code. Among these options, we keep the ones that that receive the least number of scores with value $x=2$, as these options result in fewer cycle-$2k$ in the MD-SC code. We continue until we reach $x=\lfloor L_2/2\rfloor$ or there is only one option left for the targeted circulant. Then, all survived options are recorded as branches of a tree, and the next targeted circulant is chosen and similarly evaluated for each branch.\vspace{-0cm}

\begin{remark}
The score value is by definition a divisor of $L_2$. Thus, $x$ is only considered for the above analysis if $L_2\overset{x}{=}0$.
Moreover, we do not continue the procedure until reaching $x=L_2$. This is because two options that receive the same number of scores with value $x$, $x\in\{x|x\in\{1,2,\dots,L_2/2\}, L_2\overset{x}{=}0\}$, receive the same number of scores with value $x=L_2$.\vspace{-0cm}
\end{remark}

 Algorithm~\ref{algo1} shows the procedure to find the best relocation options. 
We consider all cycles-$k$ in $\mathbf{H}_\textnormal{SC}$ that visit circulants in the middle replica. We call the cycles for which IRC holds the active cycles and the rest as the inactive cycles.
 We highlight three points here: (1) The targeted circulant $\mathcal{C}_{i_v,j_v}$ is chosen from $\{\mathcal{C}_{i,j}|\mathcal{C}_{i,j}\neq\mathbf{0}\text{ and }M(\mathcal{C}_{i,j})=0\}$ to increase the MD coupling. (2) The most problematic circulant is the one that is visited by the most active cycles. (3) Each active/inactive cycle that visits $\mathcal{C}_{i_v,j_v}$ (has $\mathcal{C}_{i_v,j_v}$ in its sequence) gives a score to each relocation option, since the status of cycles-$k$ (being active or inactive) changes by relocations.\vspace{-0cm}

\setlength{\textfloatsep}{10pt}
\begin{algorithm}[t]
\caption{{Score Voting Algorithm for Relocation}}\label{algo1}
\begin{algorithmic}[1]
\State \textbf{inputs:} targeted circulant $\mathcal{C}_{i_v,j_v}$, $k$, $[M(\mathcal{C}_{i,j})]$, $d$, and $L_2$.
\State Find $\Psi$, the set of all active/inactive cycles-$k$ that have $\mathcal{C}_{i_v,j_v}$ in their sequences.
\For{\text{each $\mathcal{O}_k\in\Psi$}}
	\For{$t \gets 0$ to $d-1$}
		\State $M(\mathcal{C}_{i_v,j_v})=t$.
		\State  $\Delta_{\mathcal{O}_k} =(-\sum_{u=1}^{k}[(-1)^{u}M(\mathcal{C}_{i_u,j_u})])_{\substack{\scriptstyle{L_2}}}$.
		\State $R(\mathcal{O}_k,t)=L_2/\gcd(L_2,\Delta_{\mathcal{O}_k})$.
	\EndFor
\EndFor
\State $\Phi=\{0,\cdots,d-1\}$.
\For{$x \gets 1$ to $\lfloor L_2/2\rfloor$}
	\If{$L_2\overset{x}{=}0$}
	\State $\Phi\gets\argmin_{t\in\Phi} |\{\mathcal{O}_k|\mathcal{O}_k\in\Psi,R(\mathcal{O}_k,t)=x\}|$.
	\EndIf
\EndFor
\State \textbf{output:} relocation options $\Phi$.
\end{algorithmic}
\end{algorithm}

Now, we are ready to describe our algorithm for constructing MD-SC codes. A solution for constructing an MD-SC code is a sequence of up to $\mathcal{T}$ relocations. Our algorithm for constructing MD-SC codes is greedy in the sense that, at each step, it chooses the relocation options that result in the least number of short cycles.
The solutions for constructing an MD-SC code are recorded in a tree structure. The root of the tree corresponds to the initial state, where $\mathbf{H}_\textnormal{SC}'=\mathbf{H}_\textnormal{SC}$ and $\mathbf{A}_t=\mathbf{0}$ for $t\in\{1,\cdots,L_2-1\}$. Other nodes correspond to one relocation each, and the path from the root node to a node at level $l$, $l\in\{1,\dots,\mathcal{T}\}$, describes a solution with $l$ relocations for constructing the MD-SC code.
At each iteration of our algorithm, we expand the tree by one level and trim the solutions that do not result in MD-SC codes with the best cycle properties amongst the solutions at that level.

\textbf{Expanding: }
At iteration $l$, we consider all nodes at level $l-1$, individually. For each node at level $l-1$, we perform the relocations described by the path from the root to the node, and form matrix $\mathbf{H}_\textnormal{SC}'$ and the auxiliary matrices, accordingly. Next, all non-zero circulants in the middle replica of $\mathbf{H}_\textnormal{SC}'$ are ranked, in a decreasing order, based on the number of active cycles-$k$ that they are visited by. Then, we target one circulant from the top of the list and find its best relocation options, by Algorithm~1. If the option ``keep in $\mathbf{H}_\textnormal{SC}'$'' is among the best options, the next problematic circulant in the sorted list is targeted. We continue this process until the most problematic circulant, such that its relocation reduces the population of short cycles, is found. Then, its best relocation options are added as children of the current node. If there is no circulant in the list such that its relocation reduces the population of short cycles, the node is not expanded.\vspace{-0cm}

\textbf{Trimming: } 
At the end of each iteration, all solutions (there is one solution per leaf node) that do not result in MD-SC codes with the least number of active cycles are trimmed. If all children of a node are trimmed, that node is trimmed as well.\vspace{-0cm}

\textbf{Termination: }
We proceed with expanding and trimming the tree of solutions, until no node is expanded in an iteration (the relocation process does not help anymore) or the  maximum density is achieved (it happens at the end of iteration $\mathcal{T}$). Then, we construct the MD-SC code according to the relocations suggested by the nodes on the path from the root to a randomly chosen, non-trimmed, leaf.\vspace{-0cm}

\begin{example}
\begin{figure}
\centering
\includegraphics[width=0.68\textwidth]{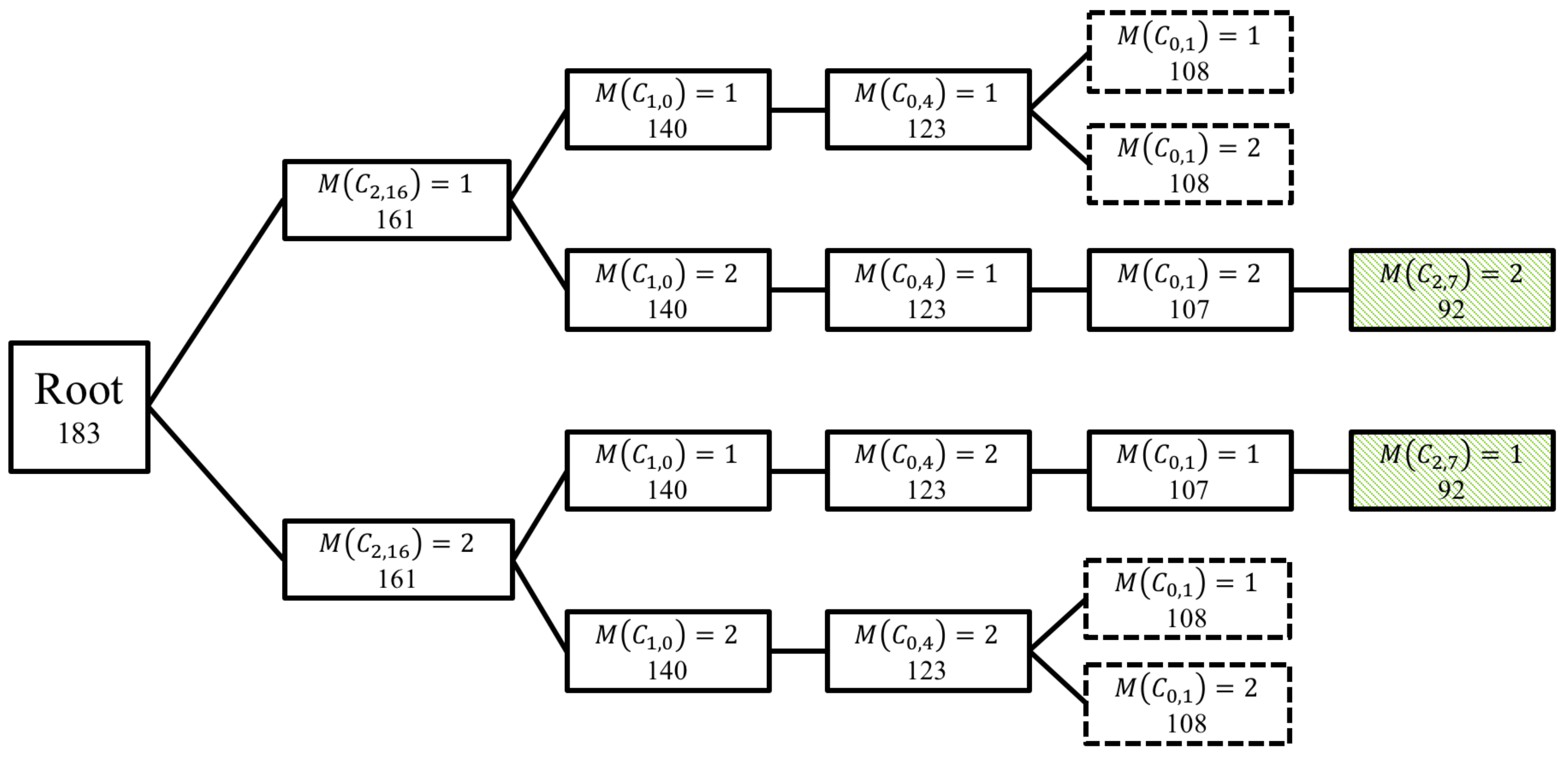}
\vspace{-0cm}
\caption{An illustration for a tree of solutions. The information associated with each node are the relocation option and the number of cycles-$6$ for the solution described by the path from the root up to this node.
The nodes with dashed borders show the trimmed solutions. 
The nodes with hatch background show the winning solutions.\vspace{-0cm}}
\label{sampletree}
\end{figure}
Fig.~\ref{sampletree} illustrates an example for the tree of solutions to construct an MD-SC code with parameters $L_2=3$, $d=3$, and $\mathcal{T}=5$ \footnote{The remaining code parameters that result in this realization are $\kappa=z=17$, $\gamma=4$, $m=1$, $L=10$, girth $6$, and OO-CPO technique is used for constructing the constituent SC codes.}. At iteration $1$, there are two winning relocation options for the targeted circulant, and they both result in $161$ cycles-$6$. At iteration $2$, each node at level $1$ is expanded to two nodes. All $4$ solutions result in $140$ active cycles-$6$. At iteration $3$, each node at level $2$ is expanded to one node. All $4$ solutions result in $123$ active cycles-$6$. At iteration $4$, the $1^\text{st}$ and $4^\text{th}$ nodes at level $3$ are expanded to two nodes each, and the $2^\text{nd}$ and $3^\text{rd}$ nodes at level $3$ are expanded to one node each. Among the $6$ solutions, two of them result in $107$ active cycles-$6$, and the remaining result in $108$ active cycles-$6$ and are trimmed. At iteration $5$ (the last iteration), each (non-trimmed) node at level $4$ is expanded to one node. The two solutions (shown with leaves that have hatch backgrounds) result in $92$ active cycles-$6$, and one of them can be chosen randomly for constructing the MD-SC code.\vspace{-0cm}
\end{example}

Algorithm~2 shows the procedure for constructing MD-SC codes.\vspace{-0cm}

\setlength{\textfloatsep}{10pt}
\begin{algorithm}
\caption{Algorithm for Constructing MD-SC Codes}\label{algo2}
\begin{algorithmic}[1]
\State \textbf{inputs:} $\mathbf{H}_\textnormal{SC}$, $k$, $L_2$, $d$, and $\mathcal{T}$.
\State \textbf{initialize:} A tree with one node (\textit{root}), $l=1$.
\State Find $\Gamma$, i.e., the set of all cycles-$k$ in $\mathbf{H}_\textnormal{SC}$ that visit the circulants in the middle replica of $\mathbf{H}_\textnormal{SC}$.
\While{$l\leq\mathcal{T}$ \textbf{and} there are nodes at level $l-1$}
\For{each \textit{node} at level $l-1$ }
\State Set $[M(\mathcal{C}_{i,j})]$ according to the relocations suggested by the path from \textit{root} to \textit{node}
\State Find status (active/inactive) of cycles-$k$ in $\Gamma$ using IRC or (\ref{mainIneffReloc}).
\State $S=\{\mathcal{C}_{i,j}|\mathcal{C}_{i,j}\in\mathbf{R}_{\lceil L/2 \rceil}\text{ and }M(C_{i,j})=0\}$.
\State Sort $S$ in a decreasing order according to the number of times they are visited by active cycles in $\Gamma$.
\State Flag$=0$.
\While{$|S|>0$ \textbf{and} Flag$=0$}
\State Select the first circulant $\mathcal{C}_{i_v,j_v}$ in $S$ for relocation.
\State Find best relocation options $\Phi$ for $\mathcal{C}_{i_v,j_v}$ by Algorithm~\ref{algo1}.
\If{$0\in\Phi$}
$S=S\setminus\mathcal{C}_{i_v,j_v}$
\Else
\State Flag$=1$.
\For{$\forall t\in\Phi$}
\State Add a child to \textit{node} with content $M(C_{i_v,j_v})=t$.
\EndFor
\EndIf
\EndWhile
\EndFor
\State Count the number of active cycles for each solution suggested by the nodes at level $l$.
\State Trim all leaves (and their parents if needed) that do not result in minimum active cycles-$k$.
\State $l=l+1$.
\EndWhile
\State Pick a random solution, set $M(C_{i,j})$ accordingly, and construct $\mathbf{H}_\textnormal{SC}^\textnormal{MD}$ using (\ref{MD_structure}).
\State \textbf{output:} $\mathbf{H}_\textnormal{SC}^\textnormal{MD}$.
\end{algorithmic}
\end{algorithm}

\vspace{-0.0cm}
\newpage
\section{Low-Latency Decoding of MD-SC Codes\vspace{-0cm}}

In this section, we present a low-latency windowed decoding for MD-SC codes. First, we describe the decoding method. Then, we provide the latency analysis of our decoder.\vspace{-0cm}

\subsection{Multi-Dimensional Windowed Decoding\vspace{-0cm}}

In this subsection, we present a low-latency windowed decoder for our MD-SC codes. Our decoder extends the well-studied windowed decoder of the conventional SC codes,  \cite{IyengarIT2012,IyengarIT2013}, to allow for low-latency decoding across multiple constituent SC codes. {Such a decoder was briefly introduced in \cite{SchmalenISTC2014}. In our paper, we thoroughly define and analyze the multi-dimensional windowed decoder which, to our knowledge, has not been done before.}

First, we recall the original windowed decoder. By construction, for a 1D-SC code, any two VNs are guaranteed not to share CNs if the replicas that they belong to are separated by more than $m$ replicas in between. As such, the farther apart two VNs are, the less likely they are to affect the decoding results of each other. This observation motivates the idea of a windowed decoder where only a subset, or a window, of VNs and CNs are considered for decoding the VNs of a replica. 

Let $W$ be the size of the window and $l$ be the window index, where $m+1 \leq W \leq L$ and $1 \leq l \leq L$. Consider the $l^\text{th}$ window. This window only considers the edges between VNs and CNs that exclusively belong to replicas $\{\mathbf{R}_{l},\mathbf{R}_{l+1} \dots,\mathbf{R}_{\min(l+W-1,L)}\}$, known as the \textit{window configuration}. We assume that replicas $\{\mathbf{R}_{1},\mathbf{R}_{2} \dots,\mathbf{R}_{l-1}\}$ have already been decoded, and their decoded values contribute to the decoding of VNs in later windows. The $l^\text{th}$ window performs decoding on its own window configuration and aims to decode the VNs in replica $\mathbf{R}_{l}$, known as the \textit{targeted VNs}. This operation is performed sequentially from the first to the $L^\text{th}$ window until the VNs of all replicas are decoded \cite{IyengarIT2012}.

We extend the idea of low-latency windowed decoding to make it applicable for decoding of the MD-SC codes. Note that each segment of an MD-SC code has the staircase structure of an SC code, see Definition 1.2. Therefore, if two VNs do not to share CNs before MD coupling, they also do not share CNs after MD coupling. Moreover, if two VNs do not to share CNs within one constituent SC code before MD coupling, any instance of these two VNs across different constituent SC codes also do not share CNs after MD coupling.

We define an MD window as a collection of several smaller (local) windows that are each defined over one segment of $\mathbf{H}_\textnormal{SC}^\textnormal{MD}$. Let $W_D$ be the size of the local windows and $l_D$ be the MD window index, where $m+1 \leq W_D \leq L$ and $1 \leq l_D \leq L$. Let $\mathbf{R}_{k}\textnormal{@}\mathbf{S}_{a,b}$ refer to the $k^\text{th}$ replica in segment $\mathbf{S}_{a,b}$ of $\mathbf{H}_\textnormal{SC}^\textnormal{MD}$. Recall that $\mathbf{S}_{(a+t)_{\substack{\scriptscriptstyle L_2}},a}=\mathbf{0}$ for $t\in\{d,\cdots,L_2-1\}$ and $a \in \{0,\dots,L_2-1\}$, which results in $\mathbf{R}_{k}$ being zero for these segments. Therefore, the local windows are only defined for the non-zero segments, and the number of local windows is $L_2d$. 

Consider the $l^{\text{th}}_D$ MD window. For this MD window, we define a local window $W^{l_D}_{a,b}$, where $0\leq a \leq L_2-1$ and $0 \leq b \leq d-1$, as the edges between VNs and CNs that exclusively belong to replicas $\{\mathbf{R}_{k}\textnormal{@}\mathbf{S}_{(a+b)_{L_2},a}, l_D \leq k \leq \min(l_D+W_D-1,L)\}.$
As such, the $l^{\text{th}}_D$ MD window is defined as the collection of local windows $W^{l_D}_{a,b}$, and we call it the \textit{multi-dimensional window configuration}.
We assume that the VNs corresponding to replicas $\{\mathbf{R}_{1},\mathbf{R}_{2} \dots,\mathbf{R}_{l_D-1}\}$ in all segments have already been decoded, and their decoded values contribute to the decoding of VNs in later MD windows. The $l_D^{\text{th}}$ window performs decoding on its own MD window configuration and aims to decode the VNs corresponding to replica $\mathbf{R}_{l_D}$ of all segments, known as the \textit{MD targeted VNs}. This operation is performed sequentially from the first to the $L^\text{th}$ MD window until all the VNs are decoded. For example in Fig.~\ref{window_fig}(a), the small rectangles represent an MD window configuration. The green columns are VNs that have already been decoded, and the blue columns are the targeted VNs.

One can view each constituent SC code as a chain of replicas. In the MD-SC code, constituent SC chains can only be connected together through their similar replicas. MD windowed decoding exploits this limited connectivity to allow for lower decoding latency, Fig.~\ref{window_fig}(b).
The structure of our MD-SC codes allows for a simpler decoder implementation. For all MD window configurations, the graphs have the same structure (edge connectivity). Thus, the same, small decoder can be used for all MD windows  and the only change across MD windows is the initial values of the VNs. This is another advantage of our deterministic construction compared to the random constructions, e.g., \cite{OhashiISIT2013,SchmalenISTC2014}, where the MD window configurations vary. 

{We now briefly mention a viable variant of the MD windowed decoding that is an interesting direction for future research. To further reduce the complexity and latency, one can limit the number of constituent SC codes that are considered in an MD window configuration to be less than $L_2$. The major challenge with this decoder is that the degree distribution of the MD window can be very different from the global degree distribution, affecting the decoding threshold of the MD windows \cite{KudekarIT2011}. Since the MD window configurations depends on the selected relocations, the score voting algorithm would need to take into account the degree distribution change. This observation requires more analysis which is left as future work. }

\begin{figure}
\centering
\begin{tabular}{cc}
\includegraphics[scale=0.25]{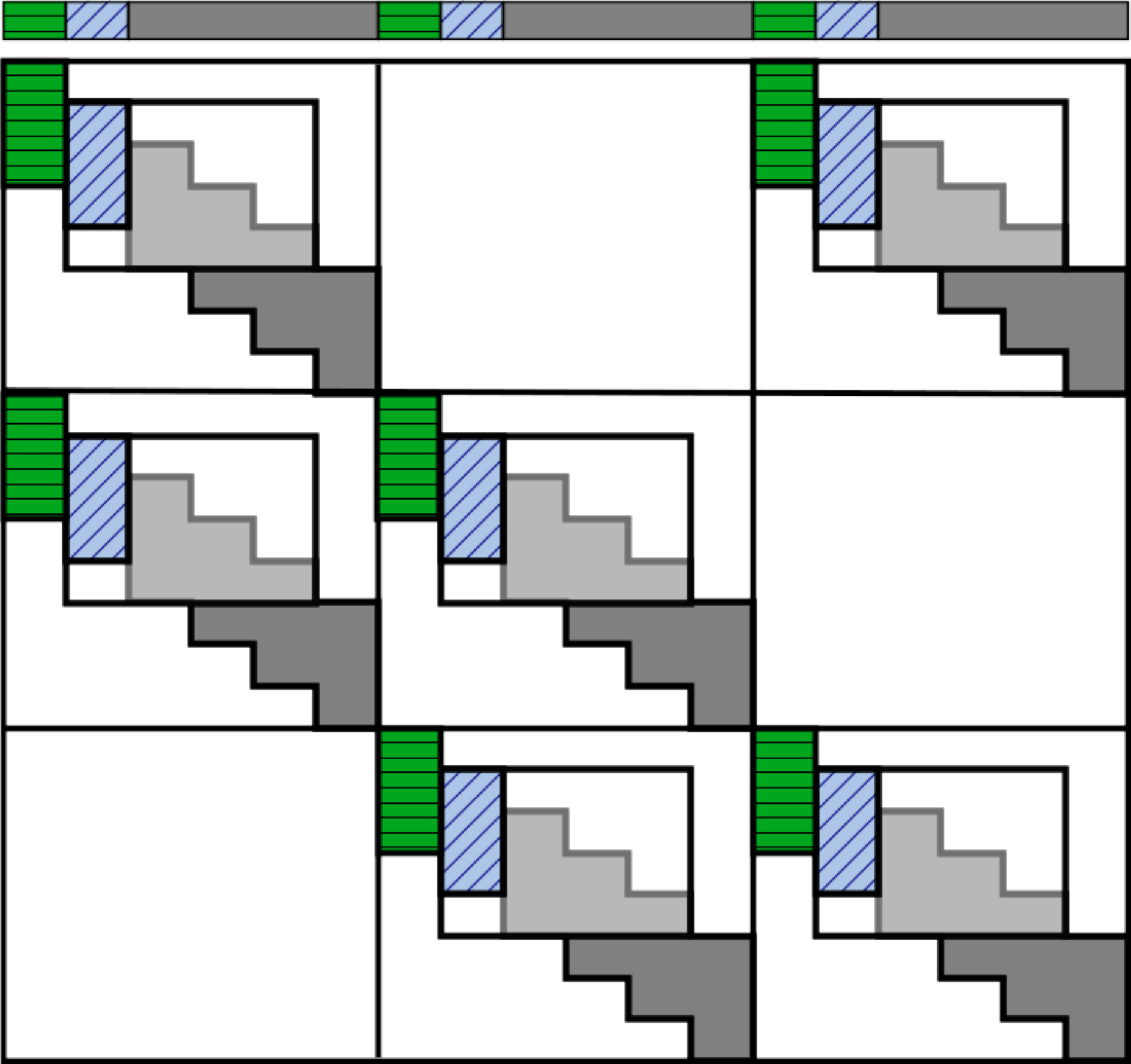}
&\hspace{1cm}
\includegraphics[scale=0.5]{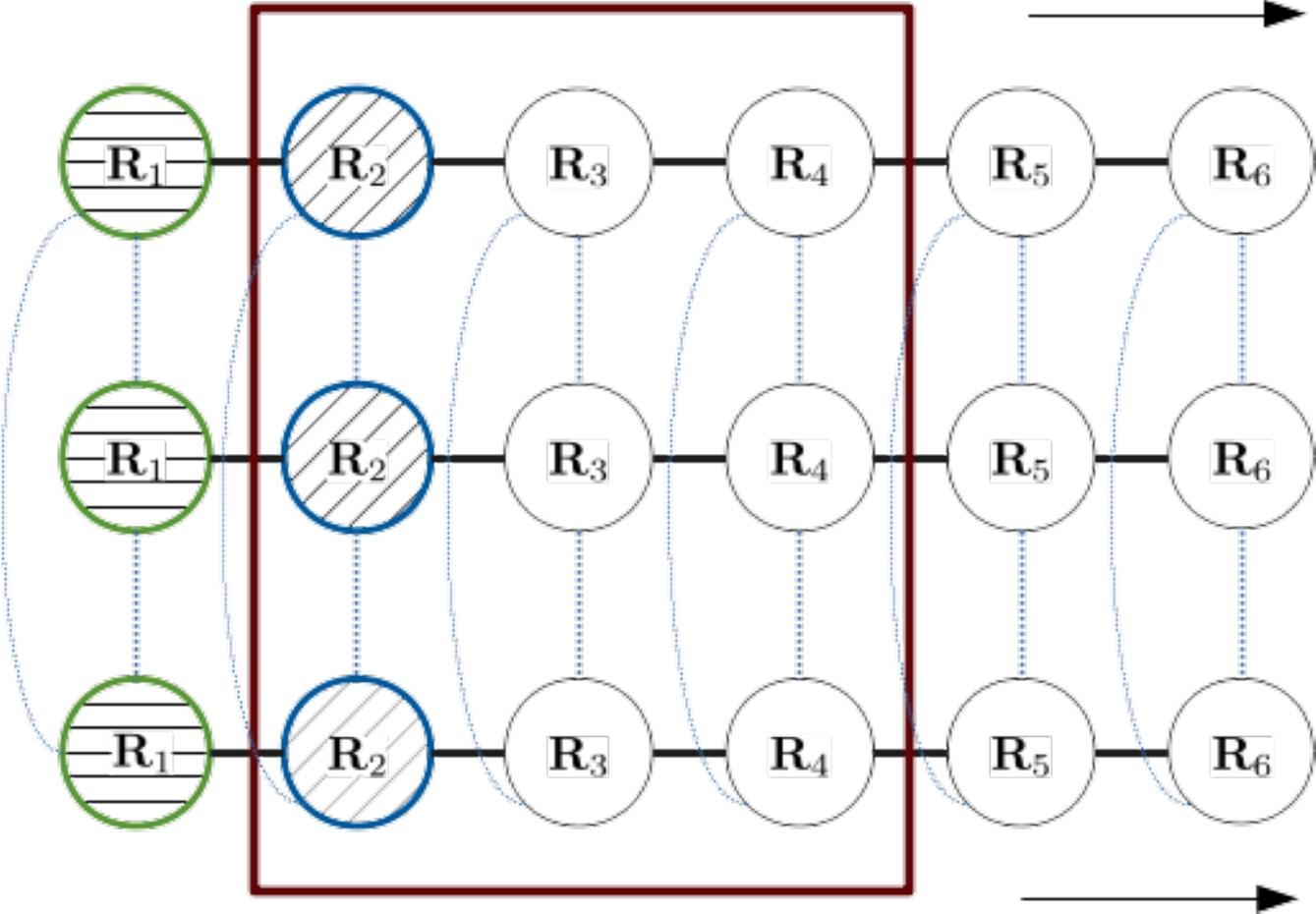}
\\
(a)&(b)\vspace{-0cm}
\end{tabular}
\caption{In both figures, the color green with horizontal lines represents the decoded VNs and the color blue with diagonal lines represents the MD targeted VNs: (a) MD window configuration for an MD-SC code with parameters $L=6$, $m=2$, $L_2=3$, $d=2$, $W_D=4$, and $l_D = 2$. (b) Each circle represents a replica and each horizontal chain is a constituent SC code. The rectangle shows the MD window configuration.
\vspace{-0cm}}
\label{window_fig}
\end{figure}
\subsection{Latency Analysis\vspace{-0cm}}

In this subsection, we provide a latency analysis of our MD windowed decoder \footnote{The presented latency analysis is inspired by the analysis in \cite{IyengarIT2012} performed for the one-dimensional windowed decoding.}. For decoding a group of VNs, we consider the decoding latency as the time the first VN is acquired until the whole VNs in that group are decoded, which is an upper bound for the latency of all VNs in the group. First, we consider the latency of a block decoder. The block decoder requires all the VNs to start decoding. Therefore, its decoding latency $T$ is given by $T= T_{rec} + T_{dec}$, where $ T_{rec}$ is the time needed to receive the full codeword and $T_{dec}$ is the time needed to decode the codeword.

We define the window latency as the time needed to decode the targeted VNs for a single MD window. This latency gives the time elapsed between successive decoding of the MD targeted VNs \footnote{This definition of latency is commonly used in the literature on windowed decoding, e.g., \cite{IyengarIT2012,IyengarIT2013}.}. We can define window latency as $T^{W_D}= T^{W_D}_{rec} + T^{W_D}_{dec}$, where $ T^{W_D}_{rec}$ is the time needed to receive the VNs in the MD window and $T^{W_D}_{dec}$ is the time needed to decode the targeted VNs. We can relate $T^{W_D}_{rec}$ and $T_{rec}$ by 
\begin{equation}
\label{latency1}
T^{W_D}_{rec} \leq \frac{(W_D+m)\kappa L_2 z}{L\kappa L_2 z}T_{rec} = \frac{(W_D+m)}{L}T_{rec},
\end{equation}
since all MD windows, except for a few trailing and leading ones, require $(W_D+m)\kappa L_2 z$ values for the VNs in their MD configurations before they can start decoding.
	
We assume that the number of iterations is the same for both the block decoder and the MD windowed decoder. For iterative decoding, the decoding time grows linearly with the number of VNs in consideration. Each MD configuration has $W_D\kappa L_2 z$ VNs, except for the last $L-W_D$ MD windows that have fewer VNs. Therefore, $T^{W_D}_{dec}$ and $T_{dec}$ are related by:
\begin{equation}
T^{W_D}_{dec} \leq \frac{W_D\kappa L_2z}{L \kappa L_2 z}T_{dec} = \frac{W_D}{L}T_{dec}
\label{latency2}.
\end{equation}
Using (\ref{latency1}) and (\ref{latency2}), $T^{W_D} \leq \frac{(W_D+m)}{L}T_{rec}+\frac{W_D}{L}T_{dec}\leq \frac{W_D+m}{L}T$. As expected, the latency reduction is similar to windowed decoding of 1D-SC codes, which shows that our MD-SC construction preserves the latency benefits of the 1D-SC codes.

\section{Simulation Results\vspace{-0cm}}

Our simulation results demonstrate the outstanding performance of our new framework for constructing MD-SC codes, and it is organized as follows: Subsections A and B are dedicated to the analysis of MD-SC codes with girth{blue}{s} $6$ and $8$, respectively. In each subsection,  we study the effect of parameters $\mathcal{T}$, $d$, and $L_2$ on the performance of MD-SC codes. Additionally, we compare the MD-SC codes constructed by our new framework with their 1D-SC counterparts (1D-SC codes having the same length and nearly the same rate as the MD-SC codes). In Subsection~C, we compare the performance of our well-designed MD-SC codes with random constructions. In Subsection~D, we evaluate the performance of the MD windowed decoding. In our simulations, we consider the AWGN channel, and we use quantized min-sum algorithm with $4$ bits and $15$ iterations for the decoding.\vspace{-0cm}

\subsection{Analysis for MD-SC Codes with Girth 6\vspace{-0cm}}

We first describe the code parameters of SC-Code~1 with girth $6$ that is used as the constituent SC code in the rest of this subsection. SC-Code~1 has parameters $\kappa=z=17$, $\gamma=4$, $m=1$, $L=10$, rate $0.74$, and length $2{,}890$ bits, and it is constructed by the OO-CPO technique \cite{EsfahanizadehTCOM2018}. The partitioning and circulant powers of SC-Code~1 are given in Appendix. The cycles of interest here have length 6, i.e., $k=6$. 

First, we consider MD-SC codes with $L_2=5$ constructed by Algorithm~2. 
Fig.~\ref{G6Sims}(a) shows the effect of increasing the MD coupling density, $\mathcal{T}$, on the population of cycles-$6$ for various MD coupling depths. The horizontal axis shows $\mathcal{T}$, and the vertical axis shows the number of active cycles-$6$. We remind that an active cycle-$k$ is a cycle-$k$ that visits circulants of the middle replica of the constituent SC code and  IRC (i.e., (\ref{mainIneffReloc})) holds for it. As we see, increasing $\mathcal{T}$ does not decrease the population of active cycles-$6$ after $18$ (resp. $23$) relocations for depth $2$ (resp., $5$), resulting in an earlier termination for the smaller depth.

\begin{figure}
    \centering
    \vspace{-0cm}
    \begin{tabular}{cc}
    \includegraphics[width=0.4\textwidth]{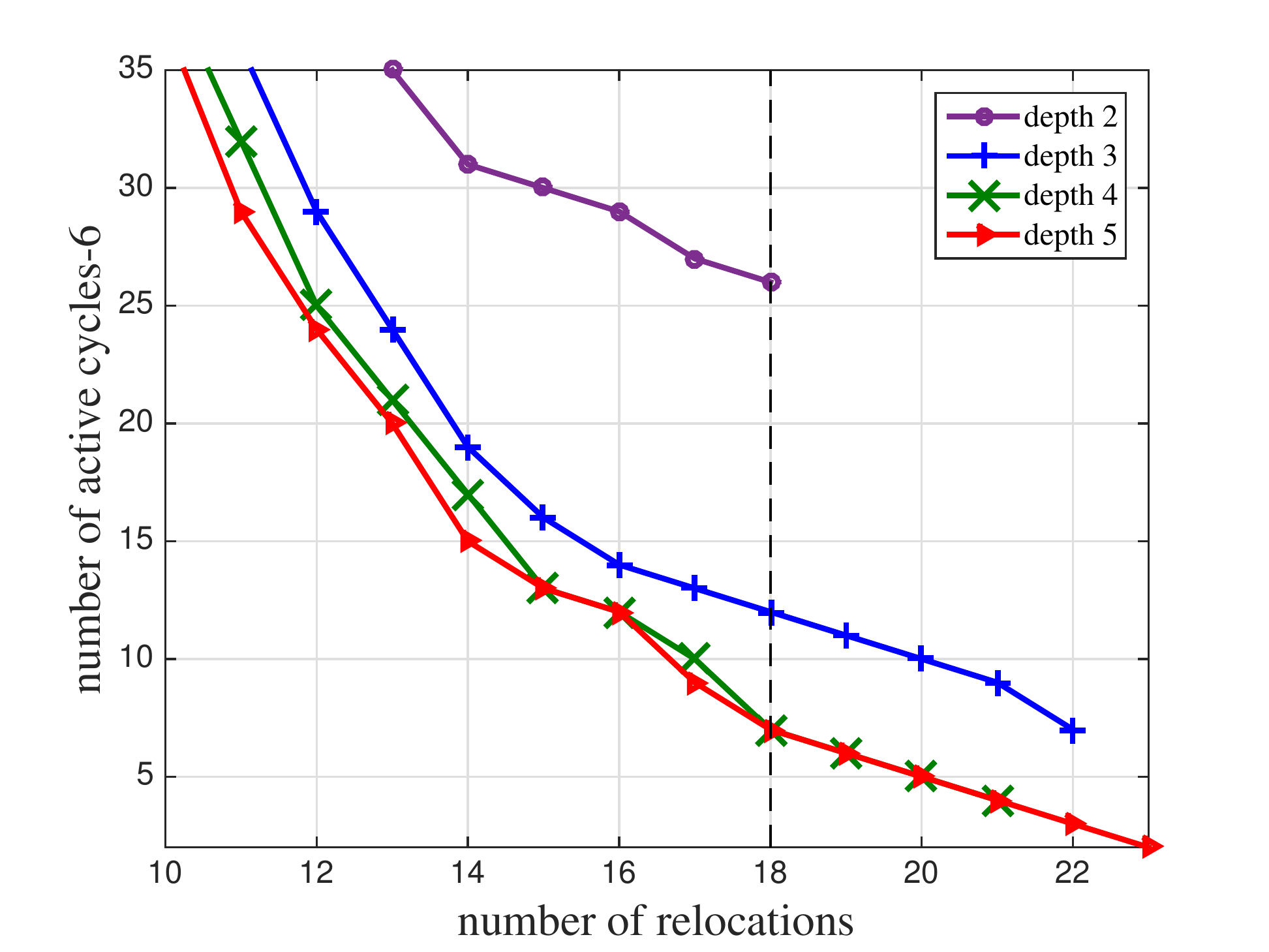}&
    \includegraphics[width=0.4\textwidth]{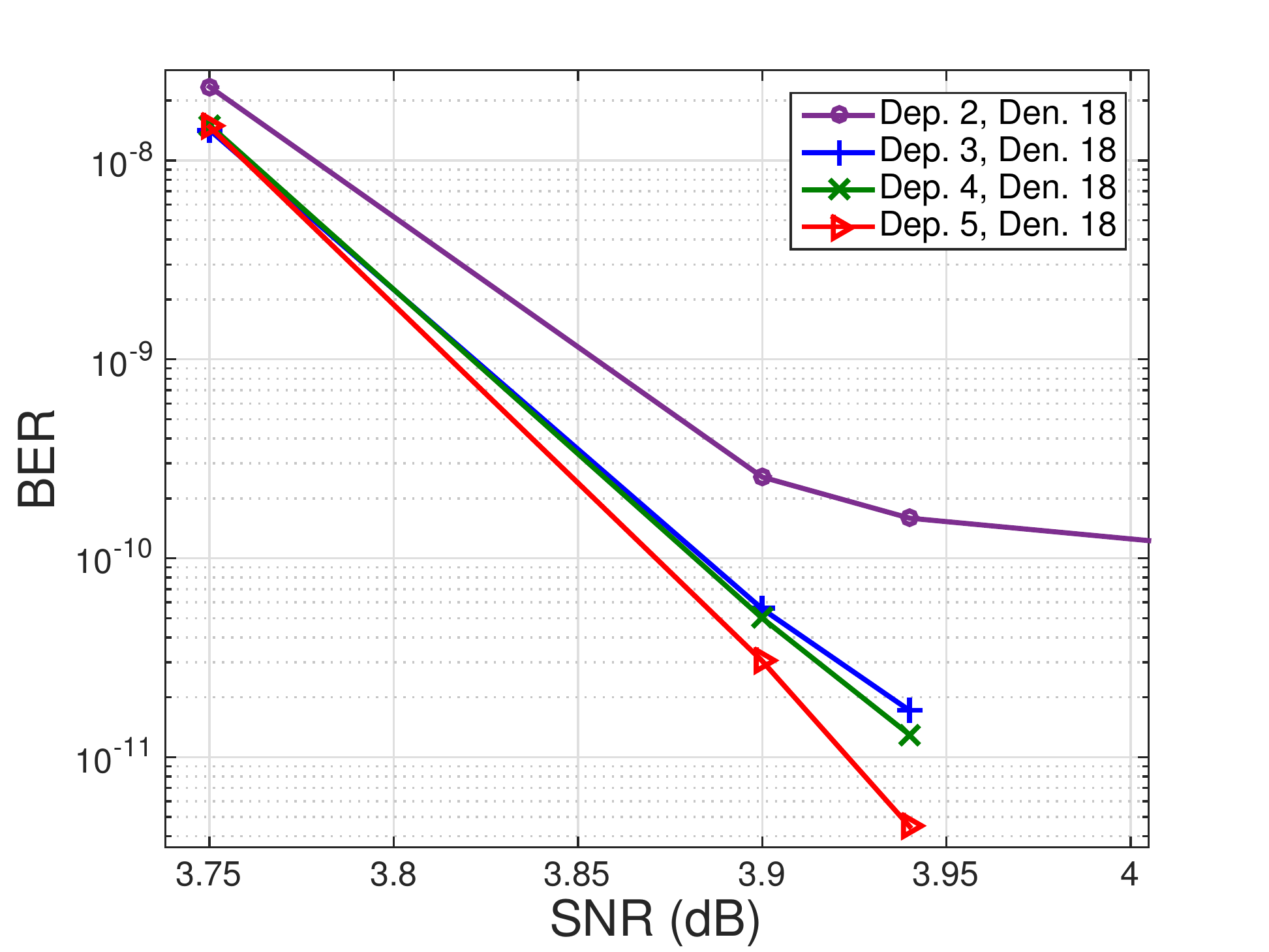}\vspace{-0cm}
    \\
     (a)&(b)\vspace{-0cm}
\end{tabular}
    \caption{MD-SC codes with SC-Code~1 as the constituent SC code and $L_2=5$: (a) The number of active cycles-$6$ for various densities and depths. Note that after a point, increasing $\mathcal{T}$ does not decrease the population of active cycles-$6$, e.g., after $18$ relocations for depth $2$, which results an early termination.
    (b) The BER performance at density $26.47\%$ and various depths.\vspace{-0cm}}
    \label{G6Sims}
\end{figure}

Table~I shows the number of cycles-$6$ for MD-SC codes with $L_2=5$, density $18$ ($26.47\%$ of circulants), and for various MD coupling depths. As we see, increasing the depth improves the cycle properties of the MD-SC codes. According to Table~I, MD-SC codes with depth{blue}{s} $4$ and $5$ have similar number of active cycles-$6$, and the small difference in the total number of cycles-$6$ is due to the different multiplicity of the active cycles-$6$ in the final MD-SC codes.
Fig.~\ref{G6Sims}(b) shows a similar comparison in terms of the BER performance. For example, at SNR$=3.94$~dB, the MD-SC code with depth $5$ shows more than $1.5$ orders of magnitude improvement in BER performance compared to MD-SC code with depth $2$.

\begin{table}
\centering
\caption{Number of cycles-$6$ for MD-SC codes with SC-Code~1 as constituent SC code, $L_2=5$, and density $26.47\%$.\vspace{-0cm}}
\begin{tabular}{|c|c|c|c|c|}
\hline
depth $d$&$2$&$3$&$4$&$5$\\
\hline
number of active cycles-$6$&$26$&$12$&$7$&$7$\\
\hline
total number of cycles-$6$&$20{,}825$&$9{,}775$&$5{,}695$&$5{,}610$\\
\hline
\end{tabular}\vspace{-0.0cm}
\end{table}

Next, we study the effect of increasing the MD coupling length, $L_2$, on the performance of MD-SC codes. We first describe the MD-SC codes and their 1D counterparts. MD-SC-Code~1 has $L_2=1$, and it is, in fact, one instance of SC-Code~1. MD-SC-Code~2 has $L_2=3$, $d=3$ (maximum depth), and $\mathcal{T}=23$ (maximum density). After reaching the maximum density, relocation does not decrease the population of the cycles of interest. SC-Code~2 is an SC code similar to SC-Code~1 but with $L=30$ (three times the coupling length of SC-Code~1); thus it has comparable length and rate to MD-SC-Code~2. MD-SC-Code~3 has $L_2=5$, $d=5$ (maximum depth), and $\mathcal{T}=23$ (maximum density). SC-Code~3 is an SC code similar to SC-Code~1 but with $L=50$; thus it has comparable length and rate to MD-SC-Code~3. The MD mapping matrices, i.e., $\mathbf{M}=[M(\mathcal{C}_{i,j})]$, for MD-SC-Codes~2 and 3 are shown below:\vspace{-0cm}

\footnotesize
\begin{equation*}
\mathbf{M}^2\hspace{-0.0cm}=\hspace{-0.0cm}\left[
\arraycolsep=3pt\def\arraystretch{1}
\begin{array}{ccccccccccccccccc}
0&1&2	&0&2&2&0&0&0&0&0&0&0&0&0&1&0\\	
1&0&0	&0&0&1&1&1&0&2&0&0&0&0&0&2&2\\	
2&0&0	&0&0&0&0&1&2&0&0&0&1&1&0&0&2\\	
0&0&0	&0&0&0&2&0&1&1&2&0&0&0&0&1&0
\end{array}\right],\hspace{0.6cm}
\mathbf{M}^3\hspace{-0.0cm}=\hspace{-0.0cm}\left[
\arraycolsep=3pt\def\arraystretch{1}
\begin{array}{ccccccccccccccccc}
0&3&3&0&2&2&0&0&0&1&0&0&0&0&1&0&0\\
1&1&3&0&0&0&0&3&0&3&4&0&0&0&0&4&2\\
0&0&0&0&0&0&0&1&1&0&0&3&4&3&0&0&4\\
0&0&0&0&3&0&0&0&4&0&0&0&2&0&0&0&0
\end{array}\right].\vspace{0cm}
\end{equation*}
\normalsize

Table~II shows the number of cycles-$6$ for SC-Code{blue}{s}~2 and 3 and MD-SC-Codes~1-3. MD-SC-Code~2 has nearly $90\%$ fewer cycles-$6$ compared to SC-Code~2, and MD-SC-Code~3 has nearly $99\%$ fewer cycles-$6$ compared to SC-Code~3. Furthermore, by increasing the number of constituent SC codes, although the overall code length increases, the number of cycles-$6$ decreases thanks to the higher amount of the MD coupling.
 
 \begin{table}
\centering
\caption{Number of cycles-$6$ for MD-SC codes with various values of $L_2$ and their 1D counterparts.\vspace{-0cm}}
\begin{tabular}{|l|l|l|l|l|}
\hline
code name&$L_2$&length&rate&cycles-$6$\\
\hline
MD-SC-Code~1 (SC-Code~1)&$1$&$2{,}890$&$0.74$&$29{,}274$\\
\hline
SC-Code~2&$1$&$8{,}670$&$0.76$&$91{,}494$\\
\hline
MD-SC-Code~2&$3$&$8{,}670$&$0.74$&$9{,}078$\\
\hline
SC-Code~3&$1$&$14{,}450$&$0.76$&$153{,}714$\\
\hline
MD-SC-Code~3&$5$&$14{,}450$&$0.74$&$1{,}700$\\
\hline
\end{tabular}\vspace{0cm}
\end{table}

{Fig.~\ref{MDvs1Dg6} compares the BER performance for our MD-SC codes and their 1D-SC counterparts.} MD-SC-Code~2 shows about $4$ orders of magnitude performance improvement compared to SC-Code~2 at SNR$=4.10$~dB. This improvement is very pronounced for MD-SC-Code~3 compared to SC-Code~3 (about $6$ orders of magnitude at SNR$=3.85$~dB). {These results demonstrate that the freedom offered by MD-SC codes is thoroughly exploited by our efficient construction framework, resulting in a large improvement in the BER performance.} One interesting observation here is that although increasing the coupling length improves the BER perfromance for 1D-SC codes, the improvement becomes incremental for large values of $L$. Therefore, adding the MD coupling to achieve an even better error correction is a promising choice.

\begin{figure}
	\begin{tabular}{cc}
	\hspace{-0.23cm}\includegraphics[width=0.4\textwidth]{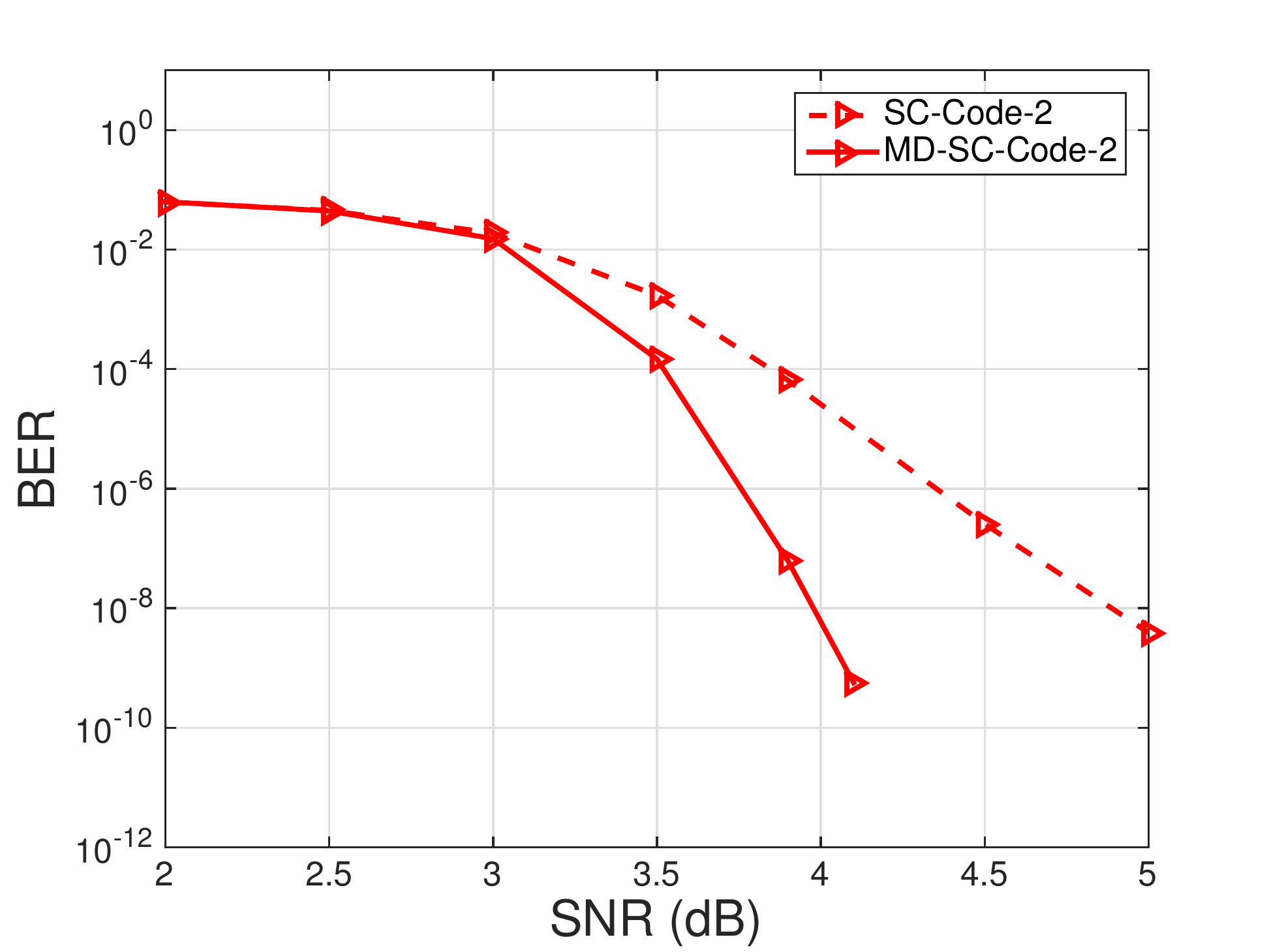}
	&
	\hspace{-0.6cm}\includegraphics[width=0.4\textwidth]{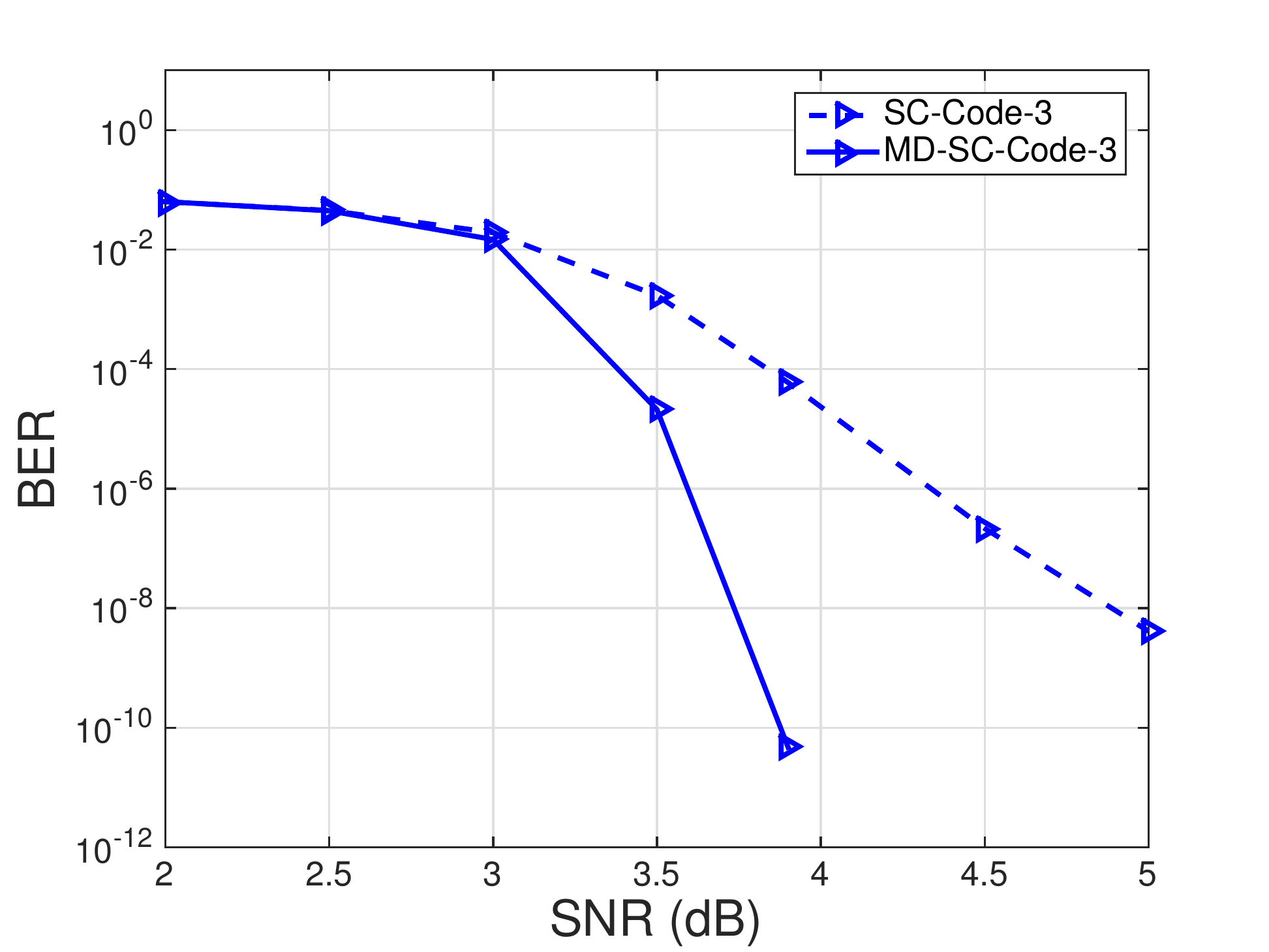}\vspace{-0.0cm}\\
		(a)&(b)\vspace{-0cm}
	\end{tabular}
	\centering
	\caption{The BER performance for MD-SC codes compared to their 1D counterparts: (a) $L_2=3$, (b) $L_2=5$.}
	\label{MDvs1Dg6}
\end{figure}

\subsection{Analysis for MD-SC Codes with Girth 8}

We first describe the code parameters of SC-Code~4 with girth $8$ that is used as constituent SC code in the rest of this subsection. SC-Code~4 has parameters $\kappa=19$, $z=23$, $\gamma=3$, $m=2$, $L=10$, rate $0.81$, and length $4{,}370$ bits, and it is constructed by the OO-CPO technique \cite{EsfahanizadehTCOM2018}. The partitioning and circulant powers of SC-Code~4 are given in Appendix. The cycles of interest here have length 8, i.e., $k=8$.

We consider MD-SC codes with $L_2=4$ constructed by Algorithm~2. 
Fig.~\ref{G8Sims}(a) shows the effect of increasing the MD coupling density, $\mathcal{T}$, on the population of cycles-$8$ for various MD coupling depths. We have two interesting observations here: First, increasing $\mathcal{T}$ does not decrease the population of active cycles-$8$ after $24$ (resp. $22$ and $21$) relocations for depth $2$ (resp., $3$ and $4$), implying that a larger depth does not necessarily result in an earlier termination. Second, for some relocations, although the population of active cycles-$8$ does not decrease, Algorithm~$2$ proceeds with relocations (for example, see relocations $18^\text{th}$ and $19^\text{th}$ in Fig.~\ref{G8Sims}(a)). This is because although these relocations do not reduce the population of the shortest cycles (cycles with length $8$ here), they reduce the population of cycles with length $2k=16$. 

\begin{figure}
    \centering
    \begin{tabular}{cc}
    \includegraphics[width=0.4\textwidth]{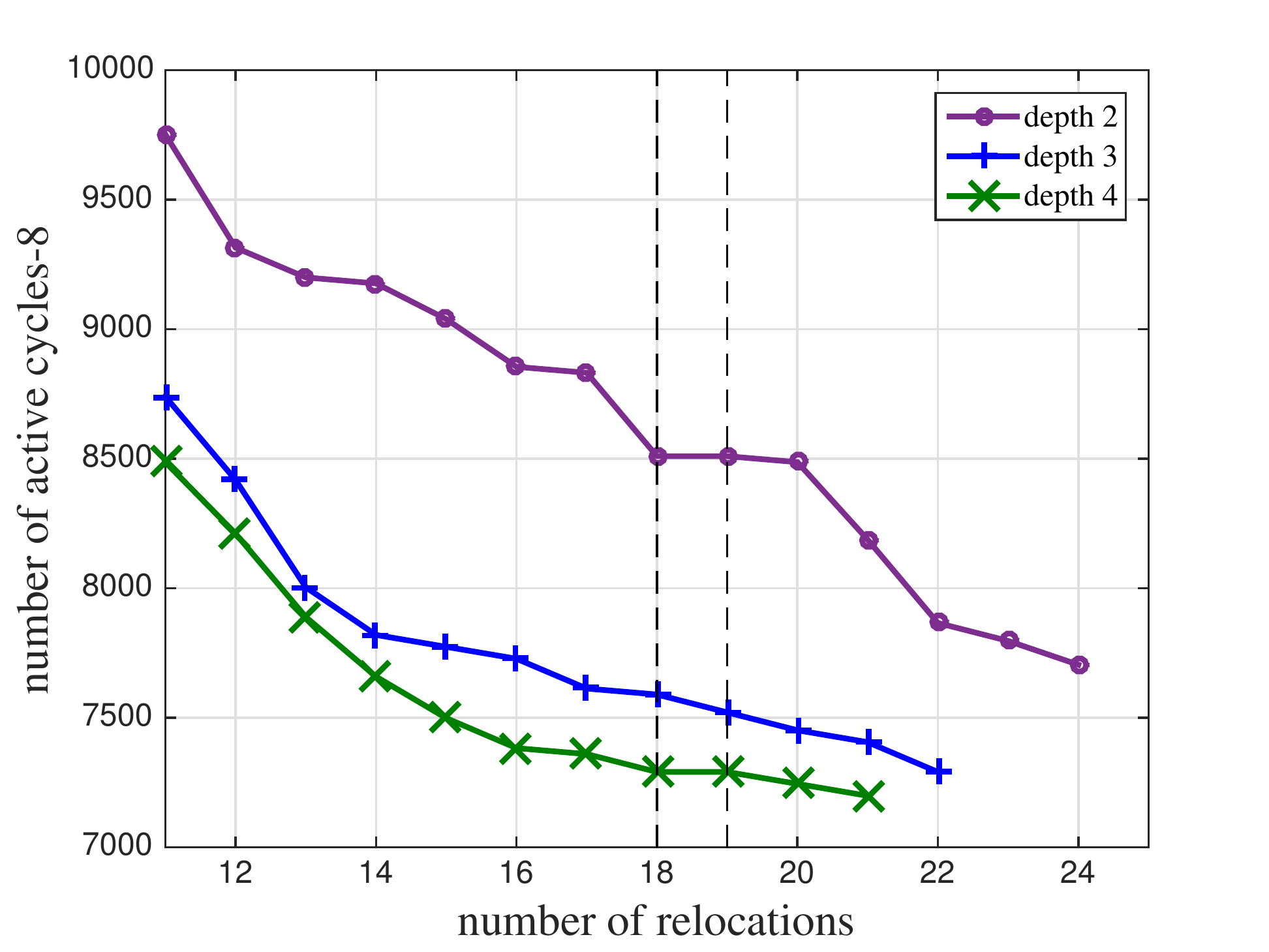}&
    \includegraphics[width=0.4\textwidth]{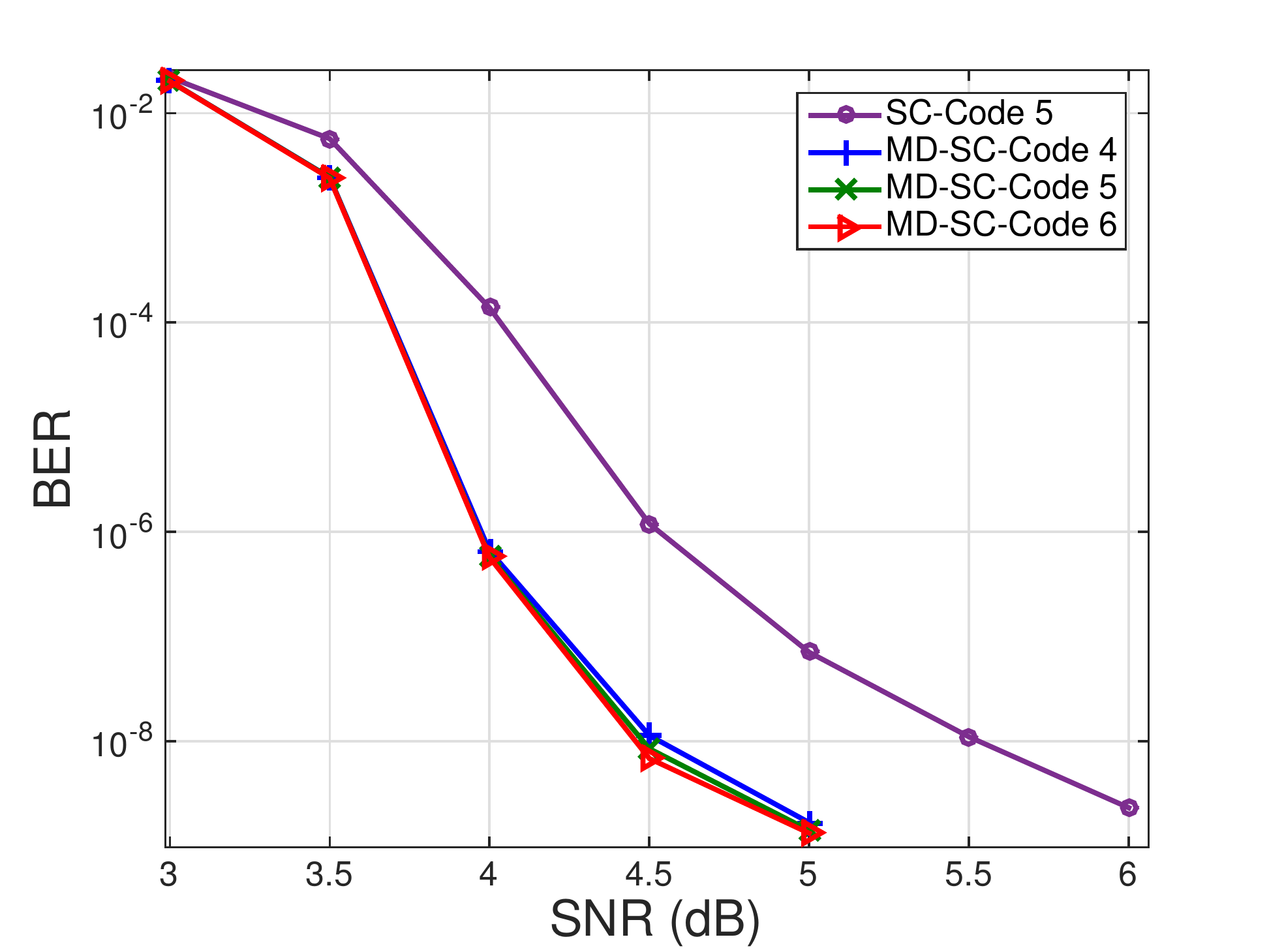}\vspace{-0.0cm}
    \\
     (a)&(b)
\end{tabular}
    \caption{MD-SC codes with SC-Code~4 as the constituent SC code and $L_2=4$: (a) The number of active cycles-$8$ for various densities and depths. (b) The BER performance for density $25\%$ and various depths along with the BER performance for the 1D-SC counterpart (SC-Code~5).}
    \label{G8Sims}
\end{figure}

Next, we study the BER performance of MD-SC codes with various depths and their 1D-SC counterpart. We first describe the codes: MD-SC-Codes~4-6 have $L_2=4$, $\mathcal{T}=19$, SC-Code~4 as their constituent SC codes, length $17{,}480$, and rate $0.81$. MD-SC-Code~4, resp. 5 and 6, have depth $2$, resp., $3$ and $4$. SC-Code~5 is an SC code similar to SC-Code~4 but with $L=40$ (four times the coupling length of SC-Code~4); thus it has comparable length and rate to MD-SC-Codes~4-6 (length $17{,}480$ and rate $0.83$)\footnote{The MD mapping matrices for MD-SC-Codes~4-6 can be found in Appendix.}.

According to Fig.~\ref{G8Sims}(b), MD-SC-Codes~4-6 show about $2$ orders of magnitude performance improvement compared to SC-Code~5 at SNR$=4.50$~dB. Table~III shows the number of cycles-$8$ for SC-Code~5 and MD-SC-Codes~4-6. MD-SC-Code~6 has nearly $82\%$ fewer cycles-$8$ compared to SC-Code~5 and nearly $15\%$ fewer cycles-$8$ compared to MD-SC-Code~4. As we see, the MD coupling considerably improves the performance of the SC codes; however, the improvement by increasing the MD coupling depth is small in this case, and thus, using a lower depth is sufficient to achieve a good error floor performance.\vspace{-0cm}

\begin{table}[!htbp]
\centering
\caption{Number of cycles-$8$ for MD-SC-Codes~4-6 and SC-Code~5.\vspace{-0cm}}
\begin{tabular}{|c|c|c|}
\hline
code name&number of active cycles-$8$&total number of cycles-$8$\\
\hline
SC-Code~5&--&$1{,}397{,}319$\\
\hline
MD-SC-Code~4&$8{,}510$&$292{,}560$\\
\hline
MD-SC-Code~5&$7{,}521$&$258{,}060$\\
\hline
MD-SC-Code~6&$7{,}291$&$249{,}320$\\
\hline
\end{tabular}\vspace{0cm}
\end{table}

\subsection{Comparison with Random Constructions}

Previous works on MD-SC codes, while promising, either consider random constructions or are limited to specific topologies. 
In this subsection, we compare our new MD-SC code construction with random constructions for connecting several SC codes together. Random constructions are inspired by \cite{TruhachevITA2012,OhashiISIT2013,SchmalenISTC2014,LiuCOMML2015}, where the purpose of random constructions is performing an ensemble asymptotic analysis over a family of the MD-SC codes. In order to perform a fair comparison, all MD-SC codes in this section have the same constituent SC code, i.e., SC-Code~6. SC-Code~6 has parameters $\kappa=17$, $z=17$, $\gamma=3$, $m=1$, $L=15$, rate $0.81$, and length $4{,}335$ bits, and it is constructed by the OO-CPO technique \cite{EsfahanizadehTCOM2018}. The partitioning and circulant powers of SC-Code~6 are given in Appendix. The cycles of interest here have length 6, i.e., $k=6$.


MD-SC-Codes~7-10 have $L_2=3$, $\mathcal{T}=9$, SC-Code~6 as their constituent SC codes, length $13{,}005$ bits, and rate $0.81$. MD-SC-Codes~7 and 8, have depths $2$ and $3$, respectively, and they are constructed by Algorithm~2 introduced in this paper\footnote{The MD mapping matrices for MD-SC-Codes~7 and 8 can be found in Appendix.}. MD-SC-Codes~9 and 10 are constructed by random relocations, and they both have depth $2$. For MD-SC-Code~9, the relocated circulants are chosen uniformly at random, and similar relocations are applied to all replicas of one constituent SC code. However, different constituent SC codes can have different relocations. MD-SC-Code~10 is constructed in a similar way to MD-SC-Code~9, but the same relocations are applied to all constituent SC codes. The later random construction has the benefit of avoiding the creation of  cycles-$4$ if the constituent SC codes do not have cycles-$4$.

\begin{table}[t]
\centering
\caption{Population of short cycles for MD-SC-Codes~7-10 (MD-SC codes constructed by different policies).\vspace{-0cm}}
\begin{tabular}{|l|c|c|c|}
\hline
code name&num. cycles-$4$&num. cycles-$6$&num. cycles-$8$\\
\hline
MD-SC-Code~7&$0$&$2{,}856$&$685{,}032$\\
\hline
MD-SC-Code~8&$0$&$0$&$643{,}110$\\
\hline
MD-SC-Code~9&$255$&$9{,}010$&$585{,}820$\\
\hline
MD-SC-Code~10&$0$&$8{,}211$&$606{,}543$\\
\hline
\end{tabular}\vspace{-0cm}
\end{table}

Table~IV shows the population of short cycles for MD-SC-Codes~7-10. As we see, MD-SC-Code~7 has $65\%$ fewer cycles-$6$ compared to MD-SC-Code~10, and they both have zero cycles-$4$. These two codes have they same structure, but the relocated circulants are chosen randomly for MD-SC-Code~7, while they are chosen to specifically reduce the number of cycles-$6$ for MD-SC-Code~10. MD-SC-Code~8, which is similar to MD-SC-Code~7 but with depth $3$, has zero cycles-$6$ and $6.1\%$ fewer cycles-$8$ compared to MD-SC-Code~7. MD-SC-Code~9 is similar to MD-SC-Code~10, but without the constraint of similar relocations for all constituent SC codes, thus it could not preserve the girth of the constituent SC codes and has cycles-$4$. Fig.~\ref{COMPandMD}(a) shows the BER performance comparison for MD-SC-Code~7 and MD-SC-Code~10. These two codes both have depth $2$ and have the MD structure described in (\ref{MD_structure}). At SNR= $6.0$~dB, MD-SC-Code~7 shows nearly $1.3$ orders of magnitude BER improvement compared to MD-SC-Code~10.\vspace{-0cm}

\begin{figure}
    \centering
    \begin{tabular}{cc}
    \includegraphics[width=0.4\textwidth]{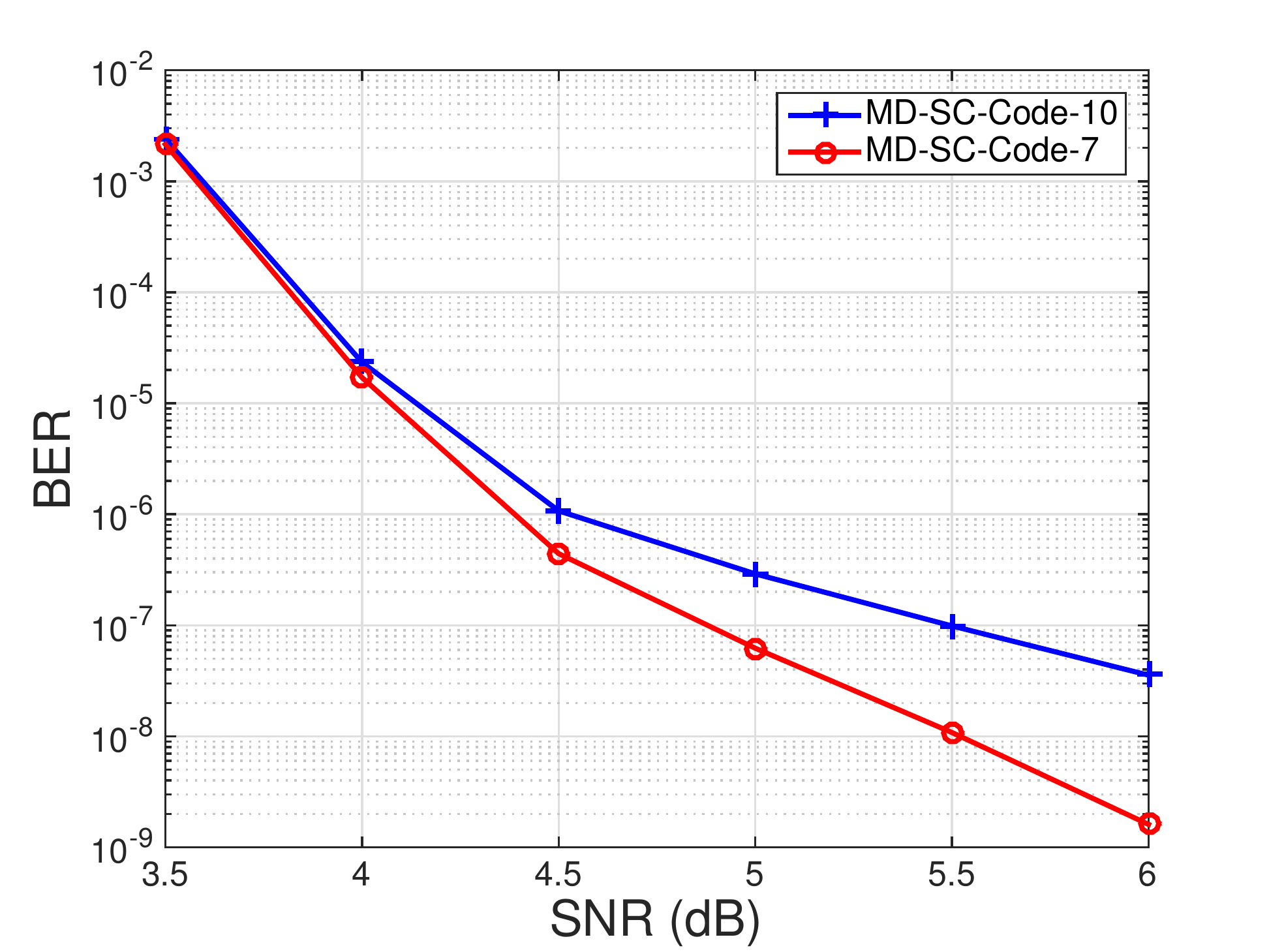}&
    \includegraphics[width=0.4\textwidth]{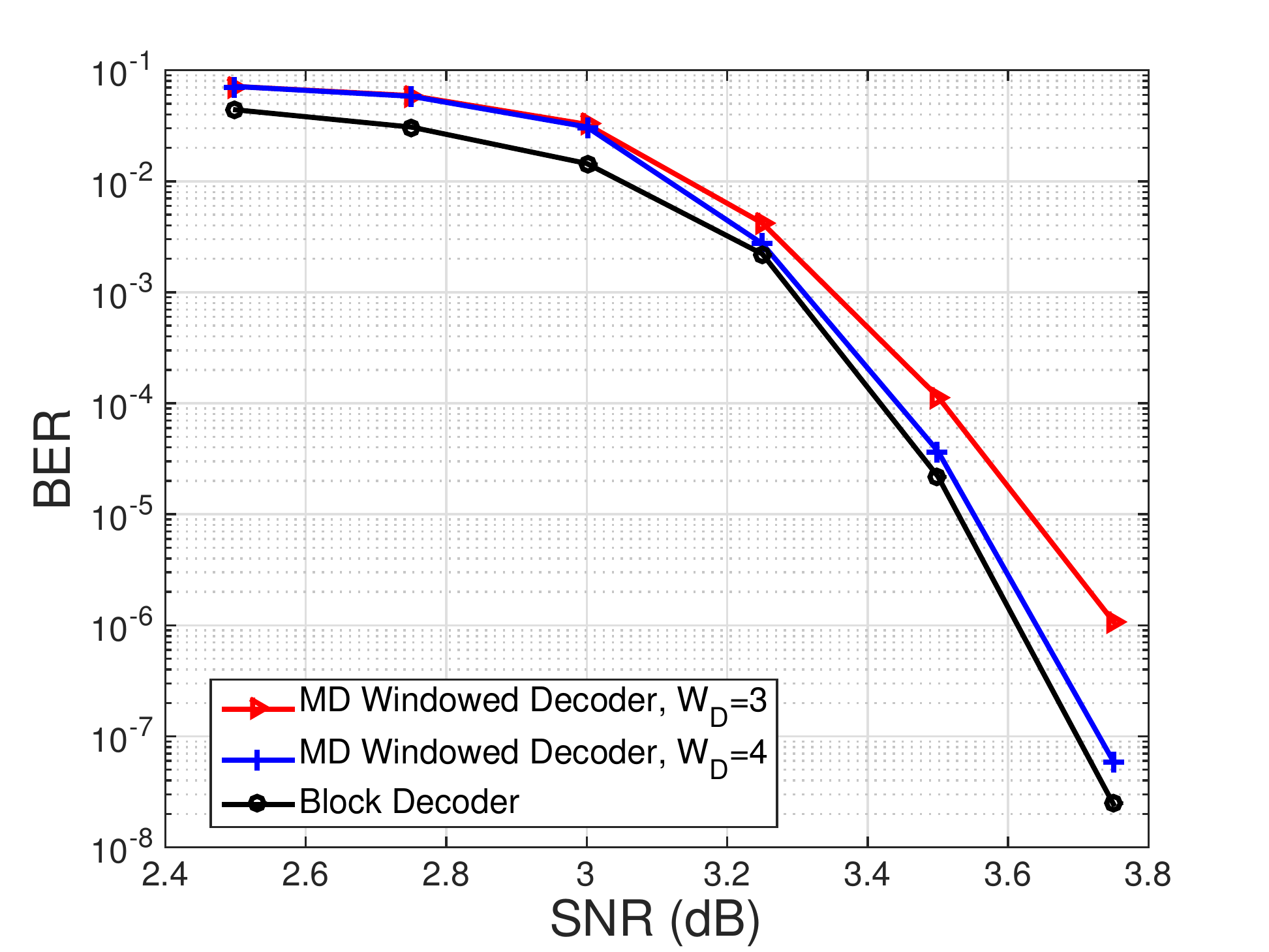}
    \\
     (a)&(b)\vspace{-0cm}
\end{tabular}
    \caption{(a) The BER performance for MD-SC codes with SC-Code~6 as the constituent SC code, $L_2=3$, density $18\%$, constructed based on a random policy and our new score-voting policy. (b) The BER performance comparison for MD windowed decoder, with MD window sizes $3$ and $4$, and block decoder for decoding an MD-SC code.\vspace{-0cm}}
    \label{COMPandMD}
\end{figure}

\subsection{Evaluation of MD Windowed Decoding}

In this subsection, we use SC-Code~1 as the constituent code and construct MD-SC-Code~11 with $L_2=5$, $d=2$, $\mathcal{T}=18$, length $14{,}450$ bits, and rate $0.74$. We evaluate the BER performance of MD-SC-Code~11 using the MD windowed decoder with MD window sizes $3$ and $4$. As a reference, we also show the BER performance using a block decoder. Both the MD windowed decoder and the block decoder use min-sum algorithm with $15$ iterations. The results are illustrated in Fig.~\ref{COMPandMD}(b). As expected, there is a slight degradation in the BER performance for windowed decoder compared to the block decoder. In addition, the degradation decreases as the MD window size increases, and it is already small for MD window size $4$.\vspace{-0cm}

\section{Appendix}

The partitioning matrix $\mathbf{PM}=[h_{i,j}]$ and circulant power matrix $\mathbf{CM}=[f_{i,j}]$, with dimensions $\gamma\times\kappa$, describe partitioning and circulant powers{, respectively}. A circulant with row group index $i$ and column group index $j$ in the block code $\mathbf{H}$ is assigned to the component matrix $\mathbf{H}_{h_{i,j}}$, and it has power $f_{i,j}$. For SC-Codes~1-3, these two matrices are given below:\vspace{-0cm}

\scriptsize
\begin{equation*}
\mathbf{PM}^1\hspace{-0.1cm}=\hspace{-0.1cm}\left[
\arraycolsep=3pt\def\arraystretch{1}
\begin{array}{ccccccccccccccccc}
0&1&0&1&0&1&0&1&0&1&0&1&0&1&0&1&1\\
1&0&1&0&1&0&1&0&1&0&1&0&1&0&1&0&0\\
0&0&0&0&0&0&0&0&1&1&1&1&1&1&1&1&1\\
1&1&1&1&1&1&1&1&0&0&0&0&0&0&0&0&0
\end{array}\right]\hspace{-0.1cm},
\mathbf{CM}^1\hspace{-0.1cm}=\hspace{-0.1cm}\left[
\arraycolsep=2pt\def\arraystretch{1}
\begin{array}{ccccccccccccccccc}
0 &10&2&8&2&0&5&7&15&0&0&0&0&10&0&0&0\\
11&15&2&14&10&3&6&7&8&9&4&11&12&8&14&10&16\\
11&2&4&12&8&11&12&9&15&4&13&5&6&1&11&13&15\\
11&3&6&9&2&16&8&4&7&10&13&16&2&5&8&6&14
\end{array}\right].
\end{equation*}
\normalsize
For SC-Codes~4-5, these two matrices are given below:
\scriptsize
\begin{equation*}
\mathbf{PM}^2\hspace{-0.1cm}=\hspace{-0.1cm}\left[
\arraycolsep=3pt\def\arraystretch{1}
\begin{array}{ccccccccccccccccccc}
0&1&1&0&1&2&0&2&2&0&1&1&0&1&2&0&2&2&2\\
1&0&0&1&0&0&1&0&0&2&2&2&2&2&1&2&1&1&1\\
2&2&2&2&2&1&2&1&1&1&0&0&1&0&0&1&0&0&0
\end{array}\right]\hspace{-0.1cm},
\end{equation*}
\begin{equation*}
\mathbf{CM}^2\hspace{-0.1cm}=\hspace{-0.1cm}\left[\hspace{-0.1cm}
\arraycolsep=2pt\def\arraystretch{1}
\begin{array}{ccccccccccccccccccc}
21&0 &16&3&19&1&0&0&21&5&0&0&1&0&9&0&16&1&0\\
 0 &11&7&3&4&5&6&7&8&9&10&11&12&13&14&15&16&17&18\\
 0&17&0 &6&8&10&12&14&16&18&20&22&1&3&5&19&9&11&13
 \end{array}\right]\hspace{-0.1cm}.
\end{equation*}
\normalsize

For SC-Code~6, these two matrices are given below:
\footnotesize
\begin{equation*}
\mathbf{PM}^3\hspace{-0.1cm}=\hspace{-0.1cm}\left[
\arraycolsep=3pt\def\arraystretch{1}
\begin{array}{ccccccccccccccccc}
1&0&1&0&1&0&1&0&1&0&1&0&1&0&1&0&1\\
0&1&0&1&0&1&0&1&0&1&0&1&0&1&0&1&0\\
1&0&0&1&0&1&1&0&0&1&1&0&1&0&1&1&0
\end{array}\right]\hspace{-0.1cm},\hspace{0.1cm}
\hspace{-0.1cm}\mathbf{CM}^3\hspace{-0.1cm}=\hspace{-0.1cm}\left[\hspace{-0.1cm}
\arraycolsep=2pt\def\arraystretch{1}
\begin{array}{ccccccccccccccccc}
0&0&2&9&0&7&4&16&2&4&2&9&0&4&13&1&1\\	
13&1&2&6&4&5&6&7&8&9&10&13&12&0&14&8&16\\	
0&2&0&0&8&10&8&14&16&1&3&5&7&15&5&5&2
 \end{array}\right]\hspace{-0.1cm}.
\end{equation*}
\normalsize

For MD-SC-Codes~4-8, the MD mapping matrices are given below:
\scriptsize
\begin{equation*}
\mathbf{M}^4\hspace{-0.0cm}=\hspace{-0.0cm}\left[
\arraycolsep=3pt\def\arraystretch{1}
\begin{array}{ccccccccccccccccccc}
0&1&1&0&1&0&0&0&0&0&1&0&0&1&1&0&0&0&0\\	
1&1&0&1&0&0&0&1&0&1&0&0&1&0&1&0&0&0&1\\
0&0&0&0&0&1&0&0&0&1&0&0&1&0&0&0&1&1&0
 \end{array}\right]\hspace{-0.1cm},
\end{equation*}
\begin{equation*}
\mathbf{M}^5\hspace{-0.0cm}=\hspace{-0.0cm}\left[
\arraycolsep=3pt\def\arraystretch{1}
\begin{array}{ccccccccccccccccccc}
0&2&1&0&1&2&0&0&0&0&1&0	&0&1&0&0&0&2&0	\\
2&0&0&1&0&0&0&1&0&0&0&0&0&0&1&0&0&0&2\\
0&1&0&0&0&1&0&0&0&1&0&0&1&1&0&1&2&0&0
 \end{array}\right]\hspace{-0.1cm},\hspace{0.1cm}
\mathbf{M}^6\hspace{-0.0cm}=\hspace{-0.0cm}\left[
\arraycolsep=3pt\def\arraystretch{1}
\begin{array}{ccccccccccccccccccc}
0&2&0&0&0&0&0&0&0&0&1&0	&0&1&1&0&0&3&0	\\
1&0&0&1&1&0&0&0&0&0&0&3	&2&0&3&0&0&0&2	\\
0&1&0&3&0&1&0&0&0&1&3&0	&0&0&0&0&2&3&0
 \end{array}\right]\hspace{-0.1cm}.
\end{equation*}
\normalsize

\scriptsize
\begin{equation*}
\mathbf{M}^7\hspace{-0.0cm}=\hspace{-0.0cm}\left[
\arraycolsep=3pt\def\arraystretch{1}
\begin{array}{ccccccccccccccccc}
1&0&1&1&0&1&0&0&1&0&0&1&0&0&1&0&0\\
0&0&0&0&0&0&0&0&0&0&0&0&0&0&0&0&0\\
0&0&0&0&0&0&0&0&0&0&0&1&0&0&0&1&0
 \end{array}\right]\hspace{-0.1cm},\hspace{0.1cm}
\mathbf{M}^8=\hspace{-0.0cm}\left[
\arraycolsep=3pt\def\arraystretch{1}
\begin{array}{ccccccccccccccccc}
2&0&1&1&0&1&0&0&0&0&2&2&0&0&1&0&0\\	
1&0&0&0&0&0&0&0&0&0&0&0&0&0&0&0&0\\
0&0&0&0&0&0&0&0&0&0&0&0&0&0&0&2&0
 \end{array}\right]\hspace{-0.1cm}.
\end{equation*}
\normalsize

\section{Conclusion}
We expanded the repertoire of SC codes by establishing a framework for MD-SC code construction with an arbitrary number of constituent SC codes and an arbitrary multi-dimensional coupling depth. For MD coupling, we rewire connections (relocate circulants) that are most problematic within each SC code. Our framework encompasses a systematic way to sequentially identify and relocate problematic circulants, thus utilizing them to connect the constituent SC codes. Our MD-SC codes show a notable reduction in the population of the small cycles and a significant improvement in the BER performance compared to the 1D setting. We also presented a windowed decoder for the MD-SC codes that exploits the locality of the constituent SC codes to attain a low decoding latency.
Two promising research directions are to investigate MD-SC codes on non-uniform channels, such as multilevel Flash and multi-dimensional magnetic recording channels, in addition to improve the presented windowed decoder by incorporating the MD coupling depth to further reduce the decoding latency and complexity.
{Furthermore, the presented methodology for constructing MD-SC codes can be extended to use CB underlying block codes that have circulants of weight $0$, $1$, or larger than $1$, and this is an interesting research direction.}\vspace{-0cm}

\section*{Acknowledgment\vspace{-0cm}}
{Research supported in part by UCLA Dissertation Year
Fellowship, a grant from ASRC-IDEMA, and NSF CCF-BSF:CIF 1718389. We would like to thank the reviewers whose suggestions helped improve and clarify this manuscript.}
\bibliographystyle{IEEEtran}
\bibliography{IEEEabrv,references}

\begin{thebibliography}{10}
\providecommand{\url}[1]{#1}
\csname url@samestyle\endcsname
\providecommand{\newblock}{\relax}
\providecommand{\bibinfo}[2]{#2}
\providecommand{\BIBentrySTDinterwordspacing}{\spaceskip=0pt\relax}
\providecommand{\BIBentryALTinterwordstretchfactor}{4}
\providecommand{\BIBentryALTinterwordspacing}{\spaceskip=\fontdimen2\font plus
\BIBentryALTinterwordstretchfactor\fontdimen3\font minus
  \fontdimen4\font\relax}
\providecommand{\BIBforeignlanguage}[2]{{%
\expandafter\ifx\csname l@#1\endcsname\relax
\typeout{** WARNING: IEEEtran.bst: No hyphenation pattern has been}%
\typeout{** loaded for the language `#1'. Using the pattern for}%
\typeout{** the default language instead.}%
\else
\language=\csname l@#1\endcsname
\fi
#2}}
\providecommand{\BIBdecl}{\relax}
\BIBdecl

\bibitem{EsfahanizadehAllerton2018}
H.~Esfahanizadeh, A.~Hareedy, and L.~Dolecek, ``Multi-dimensional
  spatially-coupled code design through informed relocation of circulants,'' in
  \emph{Proc. Annual Allerton Conf. Commun., Control and Comp.}, Oct. 2018, pp.
  695--701.

\bibitem{EsfahanizadehNVMW2019}
H.~{Esfahanizadeh}, A.~{Hareedy}, and L.~{Dolecek}, ``Multi-dimensional
  spatially-coupled code design with improved cycle properties,'' in
  \emph{Proc. Annual Non-Volatile Memories Workshop}, Mar. 2019.

\bibitem{FelstromIT1999}
A.~J. Felstrom and K.~S. Zigangirov, ``Time-varying periodic convolutional
  codes with low-density parity-check matrix,'' \emph{IEEE Trans. Inf. Theory},
  vol.~45, no.~6, pp. 2181--2191, Sep. 1999.

\bibitem{TannerIT2004}
R.~M. Tanner, D.~Sridhara, A.~Sridharan, T.~E. Fuja, and D.~J. Costello,
  ``{LDPC} block and convolutional codes based on circulant matrices,''
  \emph{IEEE Trans. Inf. Theory}, vol.~50, no.~12, pp. 2966--2984, Dec. 2004.

\bibitem{LentmaierIT2010}
M.~Lentmaier, A.~Sridharan, D.~J. Costello, and K.~S. Zigangirov, ``Iterative
  decoding threshold analysis for {LDPC} convolutional codes,'' \emph{IEEE
  Trans. Inf. Theory}, vol.~56, no.~10, pp. 5274--5289, Oct. 2010.

\bibitem{KudekarIT2013}
S.~{Kudekar}, T.~{Richardson}, and R.~L. {Urbanke}, ``Spatially coupled
  ensembles universally achieve capacity under belief propagation,'' \emph{IEEE
  Trans. Info. Theory}, vol.~59, no.~12, pp. 7761--7813, Dec. 2013.

\bibitem{EsfahanizadehTCOM2018}
H.~{Esfahanizadeh}, A.~{Hareedy}, and L.~{Dolecek}, ``{Finite-Length
  Construction of High Performance Spatially-Coupled Codes via Optimized
  Partitioning and Lifting},'' \emph{IEEE Trans. Commun.}, vol.~67, no.~1, pp.
  3--16, Jan. 2019.

\bibitem{MitchellISIT2017}
D.~G.~M. {Mitchell} and E.~{Rosnes}, ``Edge spreading design of high rate
  array-based {SC-LDPC} codes,'' in \emph{Proc. IEEE Int. Symp. Inf. Theory},
  Jun. 2017, pp. 2940--2944.

\bibitem{BeemerISITA2016}
A.~{Beemer} and C.~A. {Kelley}, ``Avoiding trapping sets in {SC-LDPC} codes
  under windowed decoding,'' in \emph{Proc. IEEE Int. Symp. Inf. Theory and Its
  Applications}, Oct. 2016, pp. 206--210.

\bibitem{TruhachevITA2012}
D.~Truhachev, D.~G.~M. Mitchell, M.~Lentmaier, and D.~J. Costello, ``New codes
  on graphs constructed by connecting spatially coupled chains,'' in
  \emph{Proc. Inf. Theory and App. Workshop}, Feb. 2012, pp. 392--397.

\bibitem{OhashiISIT2013}
R.~Ohashi, K.~Kasai, and K.~Takeuchi, ``Multi-dimensional spatially-coupled
  codes,'' in \emph{Proc. IEEE Int. Symp. Inf. Theory}, Jul. 2013, pp.
  2448--2452.

\bibitem{TruhachevICC2012}
D.~{Truhachev}, D.~G.~M. {Mitchell}, M.~{Lentmaier}, and D.~J. {Costello},
  ``Connecting spatially coupled {LDPC} code chains,'' in \emph{Proc. IEEE Int.
  Conf. Commun.}, Jun. 2012, pp. 2176--2180.

\bibitem{OlmosITW2013}
P.~M. Olmos, D.~G.~M. Mitchell, D.~Truhachev, and D.~J. Costello, ``A finite
  length performance analysis of {LDPC} codes constructed by connecting
  spatially coupled chains,'' in \emph{Proc. IEEE Inf. Theory Workshop}, Sep.
  2013, pp. 1--5.

\bibitem{SchmalenISTC2014}
L.~Schmalen and K.~Mahdaviani, ``Laterally connected spatially coupled code
  chains for transmission over unstable parallel channels,'' in \emph{Proc.
  Int. Symp. Turbo Codes Iterative Inf. Processing}, Aug. 2014, pp. 77--81.

\bibitem{LiuCOMML2015}
Y.~Liu, Y.~Li, and Y.~Chi, ``Spatially coupled {LDPC} codes constructed by
  parallelly connecting multiple chains,'' \emph{IEEE Commun. Letters},
  vol.~19, no.~9, pp. 1472--1475, Sep. 2015.

\bibitem{TanakaWCNC2017}
R.~Tanaka and K.~Ishibashi, ``Robust coded cooperation based on
  multi-dimensional spatially-coupled repeat-accumulate codes,'' in \emph{Proc.
  IEEE Wireless Commun. and Networking Conf.}, Mar. 2017, pp. 1--6.

\bibitem{AliWCL2018}
I.~{Ali}, H.~{Lee}, A.~{Hussain}, and S.~{Kim}, ``Protograph-based folded
  spatially coupled {LDPC} codes for burst erasure channels,'' \emph{IEEE
  Wireless Commun. Letters}, pp. 1--1, 2018.

\bibitem{WangIT2013}
Y.~{Wang}, S.~C. {Draper}, and J.~S. {Yedidia}, ``Hierarchical and high-girth
  {QC} {LDPC} codes,'' \emph{IEEE Trans. Info. Theory}, vol.~59, no.~7, pp.
  4553--4583, Jul. 2013.

\bibitem{HareedyGC2018}
A.~{Hareedy}, H.~{Esfahanizadeh}, A.~{Tan}, and L.~{Dolecek},
  ``Spatially-coupled code design for partial-response channels: Optimal
  object-minimization approach,'' in \emph{Proc. IEEE Global Commun. Conf.},
  Dec. 2018, pp. 1--7.

\bibitem{Richardson2003}
T.~Richardson, ``Error floors of {LDPC} codes,'' in \emph{Proc. Annual Allerton
  Conf. Commun., Control and Comp.}, Oct. 2003, pp. 1426--1435.

\bibitem{DolecekIT2010}
L.~Dolecek, Z.~Zhang, V.~Anantharam, M.~J. Wainwright, and B.~Nikolic,
  ``Analysis of absorbing sets and fully absorbing sets of array-based {LDPC}
  codes,'' \emph{IEEE Trans. Inf. Theory}, vol.~56, no.~1, pp. 181--201, Jan.
  2010.

\bibitem{BanihashemiIT2014}
M.~{Karimi} and A.~H. {Banihashemi}, ``On characterization of elementary
  trapping sets of variable-regular {LDPC} codes,'' \emph{IEEE Trans. Info.
  Theory}, vol.~60, no.~9, pp. 5188--5203, Sep. 2014.

\bibitem{BanihashemiIT2016}
Y.~{Hashemi} and A.~H. {Banihashemi}, ``New characterization and efficient
  exhaustive search algorithm for leafless elementary trapping sets of
  variable-regular {LDPC} codes,'' \emph{IEEE Trans. Info. Theory}, vol.~62,
  no.~12, pp. 6713--6736, Dec. 2016.

\bibitem{SmarandacheIT2012}
R.~{Smarandache} and P.~O. {Vontobel}, ``Quasi-cyclic {LDPC} codes: Influence
  of proto- and tanner-graph structure on minimum hamming distance upper
  bounds,'' \emph{IEEE Trans. Info. Theory}, vol.~58, no.~2, pp. 585--607, Feb.
  2012.

\bibitem{BattaglioniTCOM2018}
M.~{Battaglioni}, M.~{Baldi}, and G.~{Cancellieri}, ``Connections between
  low-weight codewords and cycles in spatially coupled {LDPC} convolutional
  codes,'' \emph{IEEE Trans. Commun.}, vol.~66, no.~8, pp. 3268--3280, Aug.
  2018.

\bibitem{MacKayIT1999}
D.~J.~C. MacKay, ``Good error-correcting codes based on very sparse matrices,''
  \emph{IEEE Trans. Inf. Theory}, vol.~45, no.~2, pp. 399--431, Mar. 1999.

\bibitem{EsfahanizadehTMAG2017}
H.~Esfahanizadeh, A.~Hareedy, and L.~Dolecek, ``Spatially coupled codes
  optimized for magnetic recording applications,'' \emph{IEEE Trans.
  Magnetics}, vol.~53, no.~2, pp. 1--11, Feb. 2017.

\bibitem{Hareedy2017ITW}
A.~Hareedy, H.~Esfahanizadeh, and L.~Dolecek, ``High performance non-binary
  spatially-coupled codes for flash memories,'' in \emph{Proc. IEEE Inf. Theory
  Workshop}, Nov. 2017, pp. 229--233.

\bibitem{FossorierIT2004}
M.~P.~C. Fossorier, ``Quasi-cyclic low-density parity-check codes from
  circulant permutation matrices,'' \emph{IEEE Trans. Inf. Theory}, vol.~50,
  no.~8, pp. 1788--1793, Aug. 2004.

\bibitem{BattaglioniArxiv2019}
\BIBentryALTinterwordspacing
M.~Battaglioni, F.~Chiaraluce, M.~Baldi, and D.~Mitchell, ``Efficient search
  and elimination of harmful objects in optimized {QC} {SC-LDPC} codes,''
  \emph{CoRR}, vol. abs/1904.07158, 2019. [Online]. Available:
  \url{http://arxiv.org/abs/1904.07158}
\BIBentrySTDinterwordspacing

\bibitem{IyengarIT2012}
A.~R. Iyengar, M.~Papaleo, P.~H. Siegel, J.~K. Wolf, A.~Vanelli-Coralli, and
  G.~E. Corazza, ``Windowed decoding of protograph-based {LDPC} convolutional
  codes over erasure channels,'' \emph{IEEE Trans. Info. Theory}, vol.~58,
  no.~4, pp. 2303--2320, 2012.

\bibitem{IyengarIT2013}
A.~R. Iyengar, P.~H. Siegel, R.~L. Urbanke, and J.~K. Wolf, ``Windowed decoding
  of spatially coupled codes,'' \emph{IEEE Trans. Inf. Theory}, vol.~59, no.~4,
  pp. 2277--2292, Apr. 2013.

\bibitem{KudekarIT2011}
S.~Kudekar, T.~J. Richardson, and R.~L. Urbanke, ``Threshold saturation via
  spatial coupling: Why convolutional {LDPC} ensembles perform so well over the
  {BEC},'' \emph{IEEE Trans. Inf. Theory}, vol.~57, no.~2, pp. 803--834, Feb.
  2011.

\end{thebibliography}

\end{document}